\documentclass[conference]{IEEEtran}

\IEEEoverridecommandlockouts

\ifCLASSINFOpdf
\else
\fi
\usepackage[cmex10]{amsmath}

\usepackage{amsfonts}
\usepackage{amssymb}

\usepackage{algorithmic,algorithm}

\usepackage{color}

\usepackage{dblfloatfix}
\usepackage{float}
\usepackage[font=small]{caption}

\usepackage{graphicx}

\graphicspath{{./figs/}}

\newtheorem{lemma}{Lemma}[section]

\newtheorem{theorem}{Theorem}[section]

\usepackage[tight,footnotesize]{subfigure}
\usepackage{url}

\usepackage{hyperref}

\begin{document}
%
\title{Joint Cache Resource Allocation and Request Routing for In-network Caching Services}



%

\author{\IEEEauthorblockN{Weibo Chu\IEEEauthorrefmark{1}, Mostafa Dehghan\IEEEauthorrefmark{2}, John C.S. Lui\IEEEauthorrefmark{3}, Don Towsley\IEEEauthorrefmark{4}, Zhi-Li Zhang\IEEEauthorrefmark{5}}

\IEEEauthorblockA{\IEEEauthorrefmark{1}Northwestern Polytechnical University, Xi'an, China}

\IEEEauthorblockA{\IEEEauthorrefmark{2}Google Inc., Cambridge, USA}

\IEEEauthorblockA{\IEEEauthorrefmark{3}The Chinese University of Hong Kong, Hong Kong}

\IEEEauthorblockA{\IEEEauthorrefmark{4}University of Massachusetts, Amherst, USA}

\IEEEauthorblockA{\IEEEauthorrefmark{5}University of Minnesota, Minneapolis, USA}

Email: wbchu@nwpu.edu.cn, mdehghan@google.com, cslui@cse.cuhk.edu.hk, towsley@cs.umass.edu, zhzhang@cs.umn.edu

}



\maketitle

\begin{abstract}
In-network caching is recognized as an effective solution to offload content servers and the network. A cache service provider (SP) always has incentives to better utilize its cache resources by taking into account diverse roles that content providers (CPs) play, e.g., their business models, traffic characteristics, preferences. In this paper, we study the cache resource allocation problem in a Multi-Cache Multi-CP environment. We propose a cache partitioning approach, where each cache can be partitioned into slices with each slice dedicated to a content provider. We propose a content-oblivious request routing algorithm, to be used by individual caches, that optimizes the routing strategy for each CP. We associate with each content provider a utility that is a function of its content delivery performance, and formulate an optimization problem with the objective to maximize the sum of utilities over all content providers. We establish the biconvexity of the problem, and develop decentralized (online) algorithms based on convexity of the subproblem. The proposed model is further extended to bandwidth-constrained and minimum-delay scenarios, for which we prove fundamental properties, and develop efficient algorithms. Finally, we present numerical results to show the efficacy of our mechanism and the convergence of our algorithms.

\end{abstract}




%
\IEEEpeerreviewmaketitle

\section{Introduction}

In recent years, we have witnessed a dramatic increase in traffic over the Internet. It was reported that global IP traffic has grown 10 times from 2007 to 2015, and it will continue to increase threefold by 2020~\cite{index2016cisco}. Among the various types of traffic generated by different applications, traffic from wireless and mobile devices accounts for a significant portion, i.e., according to~\cite{cisco2017global} global mobile and wireless data traffic in 2016 amounted to 47 exabytes per month, that is 49\% of the total IP traffic.

Current Internet faces significant challenges in serving this ``Big Data'' traffic. The host-to-host communication paradigm makes it rather inefficient to deliver content to geographically distributed users due to repeated transmissions of content, which results in unnecessary bandwidth wastage and prolonged user-perceived delays. The connection-oriented communication model also provides little or poor support for user mobility -- an important feature of future networks.

The overwhelming data traffic and limitations of the current Internet has led to a call for content-oriented networking solutions. Examples include CDNs (Content Delivery Networks) and ICNs~\cite{jacobson2009networking} (Information-Centric Networks). Both advocate caching (either at network edge or network-wide) as part of network infrastructure, where content can be opportunistically cached so as to bring significant benefits such as bandwidth saving, short delays, server offloading. Due to its fundamental role in global content delivery, and the fact that cache storages are always scarce as compared to the amount of content transmitted over the Internet, how to efficiently utilize cache resources becomes a significant research topic. A furry of recent studies focus on this area, such as modeling and characterizing caching dynamics~\cite{muscariello2011bandwidth, zhang2013caching}, design and performance evaluation of caching mechanisms~\cite{garetto2016unified, carofiglio2011experimental}, to name a few.

In this paper, we envision that besides maximizing cache performance (measured in hit rate or miss probability) as most previous work concentrated, we also study how cache resources in network should be utilized in a way that better supports general management purposes (e.g., QoS, fairness). Particularly, since content providers (CPs) have business relations with cache providers, a cache provider always has incentives to utilize its cache resources fully by taking into account diverse roles that CPs play in the market, e.g., their heterogeneous traffic characteristics, business models, QoS requirements. Baring this in mind, in this paper we study the problem of allocating cache resources among multiple content providers.
We consider the problem in a ``Multi-CP Multi-Cache'' environment, where there are multiple cache resources distributed at different network locations serving user requests from multiple content providers. This is exactly the same setting for a variety of networking applications, such as CDNs, wireless/femtocell networks, web-cache design, and most recently, ICNs. 
Since there are multiple paths between each content provider and its end-users through caches, it naturally leads to a problem of jointly optimizing cache resource allocation and request routing. However, achieving system optimum by this joint optimization with the objectives of, e.g., maximizing network utility, poses a significant challenge since the problem is inherently combinatorial and NP-hard~\cite{dehghan2017complexity, ioannidis2017jointly, yeh2014vip}, and thus some optimization algorithms are needed to solve these problems efficiently (with low complexity) and practically (in a decentralized manner).

In this work, we propose a joint cache partitioning and cache-level content-oblivious request routing scheme, where we allow a cache provider to partition its caches into slices with each slice dedicated to a content provider, and each content provider routes its requests to caches it connects so to maximize its own utility. Note that there are two advantages of the proposed scheme: 1) cache partitioning restricts content contention for cache space into partitions for each CP, and hence it decouples the interactions among them and also provides a natural means for the cache manager to tune the performance for each CP; 2) besides its simplicity due to content-obliviousness (less state), cache-level request routing provides a unified request pattern seen by caches, which leads to nice properties, i.e., the hit probability of each content is solely affected by allocated cache amount, and the hit rate of each CP is linear in traffic volume directed to caches. Overall, our scheme is easy-to-implement and is suitable for cache resource management.

To abstract business relations between content providers and a cache provider, we associate with each CP a utility that is a function of its content delivery performance. We formulate an optimization problem in which the objective is to maximize the weighted sum of utilities over all content providers through proper cache partitioning and request routing. We prove that the formulated problem has a biconvexity structure, and hence can be effectively solved by existing algorithms~\cite{gorski2007biconvex}. We further prove that, with our proposed routing scheme, the optimal solution to the formulated problem has a special request routing configuration, i.e., all requests of each CP are directed to one cache it connects. This property together with the convexity of the resource-allocation subproblem makes it possible to design decentralized (online) algorithms to achieve optimum.

To illustrate that our model actually provides a general framework for cache resource allocation, we extend it to bandwidth-constrained and delay optimization scenarios, where there are bandwidth limitations between caches and content providers, and where the goal is to optimize content delivery latency. We formulate optimization problems for the two scenarios, and establish the same biconvexity property. In addition, we discover interesting phenomena, i.e., under bandwidth limitation the optimal solution is the one such that each CP directs its requests to at most one cache at the volume less than the maximum volume, and it either does not direct or directs requests at the maximum volume to the other caches. Based on these fundamental properties, efficient algorithms can be devised.

In summary, we make the following contributions:

\begin{itemize}

\item We propose a joint cache partitioning and cache-level content-oblivious request routing scheme in the context of multiple content providers and multiple caches, and formulate a utility-based optimization framework for cache resource management.

\item We prove fundamental properties of the formulated problem, obtain its optimal routing structure, and then develop decentralized algorithms.

\item Using utility-based framework, we further consider bandwidth-constrained and delay optimization scenarios. We formulate optimization problems for the two extensions, show that they also have nice properties which lead to efficient algorithms design.

\item We perform numerical studies to validate the efficacy of our mechanism, and demonstrate convergence of the proposed decentralized algorithms to optimal solution.
\end{itemize}

The remainder of this paper is organized as follows. We review related work in Section~\ref{sec:related_work}.  Section~\ref{sec:problem_setting} describes problem setting and basic model. In Section~\ref{sec: problem_formulation} we formulate the joint cache resource allocation and request routing problem, prove its fundamental properties by analyzing its problem structure. In Section~\ref{sec:decentralized_mechansims} we develop decentralized (online) algorithm for implementing utility-maximizing cache allocation. Section~\ref{sec:evaluation} presents numerical results and Section~\ref{sec:discussion} discusses future research directions. We conclude the paper in Section~\ref{sec:conclusion}.


\section{Related Work}
\label{sec:related_work}

The issue of cache resource allocation and management has been extensively studied in the context of CPU and memory caches (i.e., see~\cite{qureshi2006utility, kim2004fair} and the references therein). Clearly, the characteristics of the cache workload and problem settings are quite different from the networking environment, so that the techniques and design choices developed therein cannot be readily applied to our problem.

In the context of web caching, Kelly et al.~\cite{kelly1999biased} proposed a biased replacement policy for web caches to implement differentiated quality-of-service (QoS) by prioritizing cache space to servers. Ko et al.~\cite{ko2003scalable} presented a scalable QoS architecture for a shared cache storage which guarantees hitrates to multiple competing classes. Lu et al.~\cite{lu2004design} implemented an architecture for supporting differentiated caching services and adopted a control-theoretical approach to manage cache resources. Feldman and Chuang~\cite{feldman2002service} proposed a QoS caching scheme that achieves service differentiation through preferential storage allocation and objects transitions across priority queues. A general cache partitioning model that integrates both QoS classes, content priority and popularity is also presented in~\cite{feng2005general}.

In recent years, a significant research effort has been dedicated to the cache resource management issue in information-centric networks. Rossi and Rossini~\cite{rossi2012sizing} proposed to allocate content storages heterogeneously across the network by considering graph-related centrality metrics. Psaras et al.~\cite{psaras2012probabilistic} proposed probabilistic caching scheme and their studies suggested to put more cache resources at the network edge. Similarly, Fayazbakhsh et al.~\cite{fayazbakhsh2013less} demonstrated through simulations that most of performance benefits can be achieved by edge caching. Wang et al.~\cite{wang2013optimal} studied the problem of optimal cache resource allocation to network nodes by formulating it as a content placement problem.

Cache resource allocation among content providers in network for management purposes (e.g., QoS, fairness) is a new research topic. Araldo et al. in~\cite{araldo2016stochastic} adopted content-oblivious cache partitioning approach to maximize the bandwidth savings provided by the ISP cache for handling content encryption. While they focused on single-cache allocation, the problem we study here is in a Multi-Cache (and possibly with multiple service providers) environment. Hoteit et al. in~\cite{hoteit2016fair} proposed a game-theoretic cache allocation approach to implement fairness among CPs. Exact traffic information of each CP is required to solve the problem. In our previous work~\cite{chu2016allocating}, we proposed a cache partitioning approach to maximize aggregate network utilities over content providers, and demonstrated that cache partitioning actually provides performance gain as compared to sharing the cache with traditional LRU policy. In this work, we extend the problem setting to a broader Multi-Cache environment with more general cache management policies (under some mild conditions), and which unavoidably involves content request routing and cache selection issues. Therefore, the problem studied in~\cite{chu2016allocating} can be regarded as a special case.

To the best of our knowledge, this is the first work that addresses the joint cache resource allocation and request routing problem in the context of Multi-CP Multi-Cache network environment.


\section{Problem Description and Basic Model}
\label{sec:problem_setting}

\subsection{Design Space of In-network Caching Systems}

When we consider in-network caching problems, two issues arise naturally: {\it{caching policy}} and {\it{routing strategy}}. 
The former refers to the rules how content objects are placed in cache storage, and the latter is how user requests are directed to caches.
Different choices of the two parameters form a design space for in-network caching systems, and finding an appropriate design choice for the best system performance is always highly desired. Yet the problem is extremely challenging due to complex dynamics of the system.
Take for example a simple caching system where there is only one cache serving user requests for content objects from one content provider (CP).
It is well known that when the request stream is stationary, the best caching policy is to hold the top $C$ most popular content objects in cache, where $C$ is the cache size; when the request stream is non-stationary, dynamic caching polices such as LRU, FIFO, etc, are preferred.

Things get complicated when there are two or more content providers.
Similar results can be obtained, i.e., static caching leads to the best system performance if the {\it{aggregate}} request stream is stationary.
However, with multiple content providers, there is a new design choice with respect to how to utilize the cache resource: {\it{partitioning}} vs {\it{non-partitioning (sharing)}}.
Unlike the way that cache resource as a whole is contended by different providers, now the cache can be divided into multiple slices and each slice can be dedicated to a content provider.
In this case, a new question arises: should we partition the cache or should it be used as a whole piece of storage?
If the answer is ``we should partition it'', then another question: how much cache resource should be allocated to each CP for the optimal system performance?
In our previous work~\cite{chu2016allocating} we have partially addressed these two questions.
Surprisingly, it has been shown that for stationary request streams, sharing the cache with traditional policies such as LRU is statistically equivalent to partitioning it into a specific way, which makes cache sharing sub-optimal as compared to cache partitioning.

An even more complicated scenario is when there are multiple caches and multiple content providers, as shown in Figure~\ref{Multi_Cache_Topology}.
In addition to caching policy, routing strategy becomes an important design parameter for such networks.
The routing strategy can be {\it{content-aware}} and {\it{content-oblivious}}, depending on whether each CP needs to maintain the information of where content objects are served.
More specifically, under content-aware routing each CP routes requests for a content to the cache that serves it, while under content-oblivious routing it simply forwards user requests to caches that are allocated to it.
Obviously, routing strategy should match caching policy so as to maximize system performance, i.e., static caching with content-aware routing leads to the best system performance for stationary request streams.
Furthermore, caches can be {\it{cooperative}} or {\it{non-cooperative}}, where cooperative caching requires state exchange (and traffic if necessary) between caches. 
``Big Cache''~\cite{Eman2017icdcs}, for example, abstracts the multiple caches into one single big storage and caches work collaboratively to improve system performance.
Finally, it is worthy to note that other factors from real world such as delay, bandwidth constraints also play important roles in designing a practical caching system, which makes the problem even more complicated.

In this paper, instead of fully exploring the design space of in-network caching systems, we mainly focus on non-cooperative Multi-Cache Multi-CP system with cache partitioning and content-oblivious routing schemes, and address the corresponding critical design problems within a general framework.

\subsection{Problem Description}

We now formally introduce our problem. Consider a network (e.g., a single Autonomous System) where there are $M$ (edge) caches that serve user requests for the content from $K$ content providers (CPs). These caches are managed by a third-party network provider, referred to as the cache manager (or cache provider) hereafter. Content providers have business relations with the cache manager and pay for cache resources. To efficiently utilize cache resources and maximize revenue, we allow the cache manager to partition its cache into multiple slices and allocate them to content providers.


Meanwhile, given cache slices, each content provider can determine the route of its user requests to these storages for its own interest.
In this work, we consider cache-level content-oblivious request routing scheme, where each CP probabilistically distributes its user requests to the allocated caches.
The highlights of this routing scheme is twofold: 1) it is simple and has low complexity (less state) and 2) the content popularity patterns seen by caches allocated to the same CP are identical, which further decreases the computational complexity of our proposed model.
In the following section, we will show with formal proof that this routing scheme also leads to nice structural properties of the optimal solution.

Given the above setting, we further associate with each content provider a utility that is a concave, increasing function of its content delivery performance. We seek answers to the following two questions: 1) how should we partition and allocate cache resources to content providers? and 2) how should user requests be distributed to maximize the overall network cache utilization efficiency, or the weighted sum of utilities over all content providers (i.e., to implement different notion of fairness among CPs)? The problem is thus a joint cache resource allocation and request routing problem.

\begin{figure}
  \centering
  \includegraphics[scale=0.42]{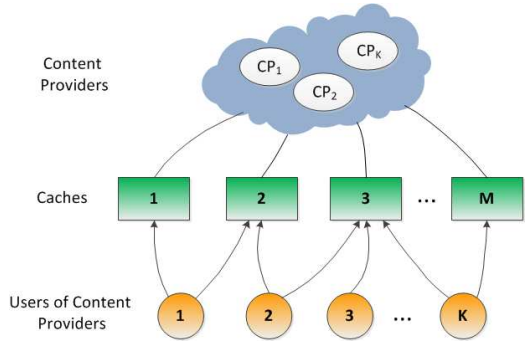}\\
  \caption{A network of multiple caches and content providers.}\label{Multi_Cache_Topology}
\end{figure}

\subsection{Cache Characteristic Time}
\label{characteristic_time}
Che et al. \cite{che2001analysis} introduced the notion of cache characteristic time. Based on this concept, the hit probability of a file\footnote{We use content and files interchangeably throughout this paper.} denoted as $o$ in an LRU cache with Poisson arrivals can be approximated by

\begin{equation}
  \label{eq:hitprob_lru}
  o(\lambda_i,T)=1-e^{-\lambda_{i}T},  
\end{equation}

\noindent where $\lambda_i$ is the request rate for file $i$, and $T$ is a constant denoting the characteristic time of the LRU cache with size $C$. $T$ can be computed as the unique solution to the equation

\begin{equation}
  \label{eq:cache_size}
  \sum_{i=1}^{N}o(\lambda_i,T)=C,
\end{equation}

\noindent where $N$ is the number of files in system. The cache hit rate is expressed as

\[h=\sum_{i=1}^{N}\lambda_io(\lambda_i,T).\]

The characteristic time approximation has proven to be an effective tool for cache performance evaluation~\cite{fricker2012versatile}\cite{leonardi2015least}. Besides LRU, it also applies to other caching policies such as FIFO, RANDOM, etc~\cite{garetto2016unified}. For example, FIFO and RANDOM have the same hit probability expressed as follows:

\begin{equation}
  \label{eq:hitprob_fifo_random}
  o(\lambda_i,T)=1-\frac{1}{1+\lambda_{i}T}.
\end{equation}

In this paper, we consider caching policies that can be well modeled using this characteristic time approximation. We rewrite $o(\lambda_i,T)$ as $o_i(\lambda_i,C)$ to explicitly denote that the hit probability is for file $i$ and it is a function of the cache size $C$. We assume that the hit rate $h$ is concave and increasing in $C$. It can be proved that this assumption holds for policies such as LRU, FIFO and Random.

\begin{theorem}
  \label{thrm:policy_concavity}
  Hit rate is a concave function of cache size for LRU, FIFO and Random policies.
\end{theorem}

\begin{proof}
  See Appendix~\ref{sec:proof_concavity}.
\end{proof}


\section{Joint Cache Resource Allocation and Request Routing}
\label{sec: problem_formulation}
In this section, we present problem formulation for the joint cache resource allocation and request routing problem.
We analyze the structural property of the problem, establish its biconvexity, and give the optimal request routing strategy. 

\subsection{Problem Formulation}
\textbf{Assumptions.} We make the following assumptions: 1) content providers serve equal-size disjoint files; 2) requests for each file arrive according to a Poisson process; and 3) each CP $k$ is associated with a utility $U_k(h_k)$ that is a concave, increasing function of its \emph{aggregate} request hit rate $h_k$.

Denote by $F_k=\{f_{1k},f_{2k},\ldots,f_{N_kk}\}$ the set of $N_k$ different files that CP $k$ serves, and $\lambda_{ik}$ the request rate for file $f_{ik}$. Define $\mathrm{\textbf{A}} = \left[ {a_{km}} \right]$ as a zero-one matrix denoting the connections between CPs and caches, with $a_{km}=1$ if requests of CP $k$ can be routed to cache $m$, and $a_{km}=0$ otherwise. Request routing strategy is described by a matrix $\mathrm{\textbf{P}} = \left[ {p_{km}} \right]$, with $0 \leq p_{km} \leq 1$ being the fraction of requests of CP $k$ routed to cache $m$. Let $C_m$ be the size of cache $m$, and $C_{km}$ be the size of cache slice that allocated to CP $k$ from cache $m$ (both measured as the number of files that can be cached). Also denote by $h_{km}$ the request hit rate to the content of CP $k$ that served by cache $m$. According to Section~\ref{characteristic_time}, we have

\begin{equation*}
\begin{aligned}
h_{km}(C_{km}, p_{km})=\sum_{i=1}^{N_k}\lambda_{ik}p_{km}o_{ik}(\lambda_{ik}p_{km},C_{km})
\end{aligned}
\end{equation*}

\noindent where $o_{ik}(\lambda_{ik}p_{km},C_{km})$ is the hit probability, in the cache $m$, for file $f_{ik}$ of CP $k$.

\begin{lemma}
  \label{lemma:hit_prob}
  As long as $p_{km} \neq 0$, we have $o_{ik}(\lambda_{ik}p_{km},C_{km})=o_{ik}(\lambda_{ik},C_{km})$.
\end{lemma}

\begin{proof}
  From equation~\eqref{eq:hitprob_lru}~\eqref{eq:cache_size} and~\eqref{eq:hitprob_fifo_random}, we can see that $\lambda_i$ and $T$ appear in their product form in these equations. As a result, for the same cache size $C$, multiplying $\lambda_i$ by $p_{km}$ will not lead to any change in the hit probability, i.e., by letting $T=T/p_{km}$ we will have the same $1-e^{-\lambda_iT}$ and $1-\frac{1}{1+\lambda_iT}$.
\end{proof}

Lemma~\ref{lemma:hit_prob} indicates that given the cache size $C$, the hit probability $o_{ik}$ is fully determined by file popularity pattern rather than the exact request volume.

Based on the above lemma, we have the following expression for $h_{km}$ when $p_{km} \neq 0$:

\[h_{km}(C_{km}, p_{km})=\sum_{i=1}^{N_k}\lambda_{ik}p_{km}o_{ik}(\lambda_{ik},C_{km})\]

\noindent Observe that when $p_{km} = 0$ we have $h_{km} = 0$, and hence the above expression indeed characterizes $h_{km}$.

The \emph{aggregate} request hit rate $h_k$ of CP $k$ is thus

\[h_{k}=\sum_{m=1}^{M}h_{km}=\sum_{m=1}^{M}\sum_{i=1}^{N_k}\lambda_{ik}p_{km}o_{ik}(\lambda_{ik},C_{km}).\]

The joint cache resource allocation and request routing problem we study can be formulated as the following utility maximization problem:

\[ \text{maximize } \sum_{k=1}^K w_kU_{k}(h_{k}(\mathrm{\textbf{C}_k}, \mathrm{\textbf{P}_k})) \]

\noindent subject to

\begin{equation} \label{eq:multi_cache_model}
\begin{aligned}
  &\sum_{k=1}^K C_{km} \leq C_m,  \quad  \forall m \\
  &\sum_{m=1}^M p_{km} = 1, \quad \forall k \\
  &0 \leq p_{km} \leq a_{km}, \quad \forall k, m \\
\end{aligned}
\end{equation}

\noindent where $\mathrm{\textbf{C}_k}=(C_{k1},C_{k2},\ldots,C_{kM})$ and $\mathrm{\textbf{P}_k}=(p_{k1},p_{k2},\ldots,p_{kM})$ denote the cache allocation and request routing for CP $k$, respectively. $w_k>0$ is the weight to CP $k$, which is chosen to reflect business preferences such
as financial incentives or legal obligations. In its most simple case where $w_k = 1$ and $U_k(h_k) = h_k$, the objective becomes that of maximizing the overall cache hit rate, which provides a measure of the cache utilization efficiency.

Note that the above formulation is a ``mixed-integer'' programming problem that is typically hard to solve. However, in practice caches are large and therefore we assume $C_{km}$ can take any real value, as the rounding error
will be negligible.

\subsection{Solution}
\label{sec:solution}
\begin{theorem}
\label{thrm:no_splitting}
The optimal solution to problem~\eqref{eq:multi_cache_model} is such that all requests of each CP are routed to one cache that it connects.
\end{theorem}

\begin{proof}
See Appendix~\ref{sec:proof_thm1}.
\end{proof}

Theorem~\ref{thrm:no_splitting} states that the optimal solution to problem~\eqref{eq:multi_cache_model} is that requests of each CP are not splitted among caches, but rather, each CP directs all its requests to one cache it connects. This result is counter-intuitive at the first glance, since to maximize utility, each content provider will try to obtain storages from multiple caches. However, this in fact makes sense in that if this is the case, then there are multiple copies of the same files cached in these storages, which results in content duplication and hence cache resource wastage. Therefore, as a contribution of this work, here we in fact prove that cache-level content-oblivious probabilistic routing is suboptimal as compared to non-probabilistic (deterministic) request routing in a multi-cache environment.

Based on Theorem~\ref{thrm:no_splitting}, Problem~\eqref{eq:multi_cache_model} can be reformulated as follows:

\[ \text{maximize } \sum_{k=1}^K w_kU_{k}(h_{k}(\mathrm{\textbf{C}_k}, \mathrm{\textbf{P}_k})) \]

\noindent subject to

\begin{equation} \label{eq:multi_cache_model_reformulated}
\begin{aligned}
  &\sum_{k=1}^K C_{km} \leq C_m,  \quad  \forall m \\
  &\sum_{m=1}^M p_{km} = 1, \quad \forall k \\
  &p_{km} \in \{0, a_{km}\}, \quad \forall k, m \\
\end{aligned}
\end{equation}

Since there is a limited number of caches and CPs, problem~\eqref{eq:multi_cache_model_reformulated} can be solved by evaluation of a series of simpler problems which are determined by the special request routing configurations given by Theorem~\ref{thrm:no_splitting}. More specifically, we can convert problem~\eqref{eq:multi_cache_model_reformulated} into a series of problems $Pr_1,Pr_2,\ldots, Pr_L$, where $L$ is the number of request routing configurations, and each $Pr_i$ corresponds to a cache resource allocation problem with one specific request routing, i.e., the one that all requests of each CP are directed to one cache that it connects. The optimal solution to problem~\eqref{eq:multi_cache_model_reformulated} is simply the one with the largest objective function value.

\begin{theorem}
\label{thrm:concavity}
Given a request routing configuration where all requests of each CP are directed to one cache that it connects, problem~\eqref{eq:multi_cache_model_reformulated} can be decomposed into a series of convex optimization problems, each for one cache.
\end{theorem}

\begin{proof}
See Appendix~\ref{sec:proof_thm2}.
\end{proof}

Based on Theorem~\ref{thrm:concavity}, we know that each $Pr_i$ can be further divided into $M$ subproblems, where each subproblem corresponds to a single-cache convex resource allocation problem. Let $B_k$ be the number of caches that CP $k$ connects. We have $L=B_1B_2\ldots B_K$ request routing configurations. Meanwhile, solving each $Pr_i$ involves solving $M$ convex single-cache resource allocation problem, and hence the time complexity of the algorithm is $O(B_1B_2\ldots B_KM)$.

The above algorithm is efficient at solving problems with small number of caches and content providers. However, as it iterates over all possible request routings, the computational time grows exponentially with the problem size. In fact, as the following theorem states, the considered problem is proved NP-hard. 

\begin{theorem}
\label{thrm:np_hardness}
Problem~\eqref{eq:multi_cache_model_reformulated} is NP-hard.
\end{theorem}

\begin{proof}
See Appendix~\ref{sec:proof_NP_hardness}.
\end{proof}

To obtain practical solution mechanisms, we further analyze structural properties of the problem.
Fortunately, the original problem turns out to be a biconvex optimization problem, as proved in Theorem~\eqref{thrm:biconvexity}.

\begin{theorem}
  \label{thrm:biconvexity}
 Problem~\eqref{eq:multi_cache_model} is a biconvex optimization problem.
\end{theorem}

\begin{proof}
  In Theorem~\ref{thrm:concavity} we show that the objective function of problem~\eqref{eq:multi_cache_model} is concave in $C_{km}$'s given a fixed $p_{km}$'s. It is also easy to see that the objective function is concave in $p_{km}$'s given a fixed $C_{km}$'s, since $h_{km}$ is linear in $p_{km}$. Therefore, the objective function is biconcave. Furthermore, as the space of the decision variables is the product of two independent convex sets, it is biconvex. As a result, problem~\eqref{eq:multi_cache_model} is a biconvex optimization problem.
\end{proof}

The following theorem states that any solution corresponding to the special routing configurations given by Theorem~\ref{thrm:no_splitting} is a {\it{partial optimum}}.

\begin{theorem}
  Let $f$ be the objective function of problem~\eqref{eq:multi_cache_model}. Denote by $X$ and $Y$ be the variable space of caching and routing, respectively. Also let $y^*$ be a routing configuration given by Theorem~\ref{thrm:no_splitting}, and $x^{*}(y^*)$ be the corresponding optimal caching. Then ($x^{*}(y^*), y^{*}$) is a partial optimum, i.e., $f(x^*(y^*), y^*) \ge f(x, y^*)$ $\forall x \in X$, and $f(x^*(y^*), y^*) \ge f(x^*(y^*), y)$ $\forall y \in Y$.
\end{theorem}

Biconvex optimization problems~\cite{gorski2007biconvex} have been extensively studied during the past few decades, and efficient algorithms (e.g., {\it{ACS}}~\cite{wendell1976minimization} for partial optimum, {\it{GOP}}~\cite{floudas1990global} for approximate global optimum) exist in literature. The good news from our numerical studies is that, even with local or partial optimum, the system performance can be significantly improved by the joint cache resource allocation and request routing optimization.


\section{Decentralized Mechanism}
\label{sec:decentralized_mechansims}


In this section, we present a decentralized mechanism to implement the optimal cache allocation and request routing. As compared to centralized solution, a decentralized mechanism is more desired in practice, i.e., when each CP does not want to reveal its utility function to the cache manager, or the cache manager cannot collect all traffic information of CPs. A decentralized mechanism also adapts to network changes naturally.
Our proposed mechanism is based on our analysis in the above section, where the system problem can be decomposed into subproblems for the caches and for the individual content providers.


\subsection{Decentralized Mechanism}
\label{sec:decentralized_algorithm}

We have shown in Section~\ref{sec: problem_formulation} the routing configurations of partial optimums. Furthermore, each of these routing configurations corresponds to a convex resource allocation problem. Let $CP(m)$ be the set of CPs that direct requests to cache $m$, and $s(k)$ be the cache to which CP $k$ forwards all its requests. Given a fixed request routing $\mathrm{\textbf{P}_k}$'s, we have the following convex resource allocation problem:

\[ \text{maximize } \sum_{k=1}^K w_kU_{k}(h_k(\mathrm{\textbf{C}_k})) \]

\noindent subject to

\begin{equation} \label{eq:mulcache_model_reformulated}
\begin{aligned}
  \sum_{k \in CP(m)} C_{km} \leq C_m,  \quad  \forall m \\
\end{aligned}
\end{equation}

\noindent where $U_{k}(h_k(\mathrm{\textbf{C}_k}))=U_{k}(h_{ks(k)}(C_{ks(k)},1))$, i.e., requests of CP $k$ are solely directed to cache $s(k)$.


The key to the decomposition of the above problem is to examine its dual. Let $\mathrm{\textbf{C}}=(\mathrm{\textbf{C}_1}, \mathrm{\textbf{C}_2}, \ldots, \mathrm{\textbf{C}_K})$ and $\boldsymbol{\lambda}=(\lambda_1, \lambda_2, \ldots, \lambda_M)$. $\boldsymbol{\lambda}$ is considered as the price vector of caches. Define Lagrangian

\begin{equation*}
\begin{aligned}
  L(\mathrm{\textbf{C}}, &\boldsymbol{\lambda})=\sum_{k=1}^K w_kU_{k}(h_k(\mathrm{\textbf{C}_k}))+\sum_{m=1}^M \lambda_m(C_m-\sum_{k \in CP(m)} C_{km}) \\
  &= \sum_{k=1}^K (w_kU_{k}(h_k(\mathrm{\textbf{C}_k}))-\sum_{m=1}^M \lambda_{m}C_{km})+\sum_{m=1}^M \lambda_mC_m \\
  &= \sum_{k=1}^K (w_kU_{k}(h_k(\mathrm{\textbf{C}_k}))-\lambda_{s(k)}C_{ks(k)})+\sum_{m=1}^M \lambda_mC_m.
\end{aligned}
\end{equation*}

\noindent The last equation holds since CP $k$ routes all its requests to cache $s(k)$ and it will not require resources from other caches.

As the first term is separable, we have $\mathop \text{maximize }\limits_{\mathrm{\textbf{C}_k}}\sum_{k=1}^K (w_kU_{k}(h_k(\mathrm{\textbf{C}_k}))-\lambda_{s(k)}C_{ks(k)})=\sum_{k=1}^K \mathop \text{maximize }\limits_{\mathrm{\textbf{C}_k}}(w_kU_k(h_k(\mathrm{\textbf{C}_k}))-\lambda_{s(k)}C_{ks(k)})$. According to \cite{palomar2006tutorial}, we can now readily formulate optimization problem for each CP $k$ as follows:

\noindent \textbf{CP $k$'s problem:}

\begin{equation}
\label{eq:CP_problem}
\begin{aligned}
  \mathop \text{maximize }\limits_{\mathrm{\textbf{C}_k}} (w_kU_{k}(h_k(\mathrm{\textbf{C}_k}))-\lambda_{s(k)}C_{ks(k)})
\end{aligned}
\end{equation}

\begin{theorem}
\label{thrm:concavity_CP_problem}
The cache resource allocation problem~\eqref{eq:CP_problem} has a unique optimal solution.
\end{theorem}

\begin{proof}
Since $h_{km}$ is concave in $C_{km}$, we know that $h_k$ is concave. Now because $U_k$ is also concave, the objective function of problem~\eqref{eq:CP_problem} is shown to be concave. Hence a unique maximizer, called the optimal solution, exists.
\end{proof}

Once each cache $m$ receives the required cache amount by content providers, its price can be adjusted as follows:

\begin{equation} \label{eq:Cache_adaptation}
\begin{aligned}
  \lambda_m^{t+1}=[\lambda_m^t+\gamma(\sum_{k \in CP(m)}C_{km}^t-C_m)]^+
\end{aligned}
\end{equation}

\noindent where $\gamma>0$ is a step size, $t$ denotes time, and $[x]^+=max\{x,0\}$.

In short, in the mechanism each content provider locally calculates its required cache amount from each cache based on their prices, and each cache adjusts its price based on the required resources from content providers. Meanwhile, as problem~\eqref{eq:CP_problem} is convex, there is no duality gap and hence the algorithm converges to the optimal solution.

It is worthy to note that the above procedure is for one specific request routing. To obtain the global optimal solution (i.e., when the problem size is small), all possible routings need to be explored. We emphasize that this process can be efficiently implemented in a parallel manner. More specifically, at each time $t$, each content provider locally calculates and maintains the information of (\emph{utility, required\_cache\_size, routing}) and broadcasts them to the caches connected. On the other hand, with these information provided by CPs, each cache perceives all possible global routings and then after calculation it   broadcasts the information of (\emph{aggregate\_utility, prices, global\_routing}) back to the CPs. The optimal prices and cache allocations will be the one with the largest (weighted) aggregate utility when all procedures converge.


\section{Bandwidth-Constrained Formulation and Delay-Oriented Optimization}
\label{sec:extensions}

\subsection{Bandwidth-Constrained Formulation}
Problem formulation~\eqref{eq:multi_cache_model} does not take into account bandwidth constraints. However, we argue that in practice, the volume of requests that each CP can forward to caches are limited, i.e., due to their bandwidth-based contracts. In this subsection, we consider such a scenario where there are bandwidth constraints between content providers and caches.

\subsubsection{Problem Formulation}
We start by giving bandwidth constraints. Let $V_{km}$ be the maximum volume of requests that CP $k$ can forward to cache $m$. We explicitly assume that $\sum_{m=1}^M V_{km}\ge \sum_{i=1}^{N_k} \lambda_{ik}, \forall k$, as all requests of each CP needs to be served by the cache network. The bandwidth constraints between content providers and caches can be expressed as the following inequation:

\[\sum_{i=1}^{N_k} \lambda_{ik}p_{km} \le V_{km}, \quad \forall k, m\]

\noindent which states that each CP cannot route requests to a cache at the volume exceeding its maximum. Incorporating this constraint into problem formulation~\eqref{eq:multi_cache_model}, we have the bandwidth-constrained formulation as follows:

\[ \text{maximize } \sum_{k=1}^K w_kU_{k}(h_{k}(\mathrm{\textbf{C}_k}, \mathrm{\textbf{P}_k})) \]

\noindent subject to

\begin{equation} \label{eq:mulcache_model_bandwidth}
\begin{aligned}
  \sum_{k=1}^K C_{km} \leq C_m,  \quad  \forall m \quad \text{(c1)}\\
  \sum_{m=1}^M p_{km} = 1, \quad \forall k \quad \text{(c2)}\\
  0 \leq p_{km} \leq a_{km}, \quad \forall k, m \quad \text{(c3)}\\
  \sum_{i=1}^{N_k} \lambda_{ik}p_{km} \le V_{km}, \quad \forall k, m \quad \text{(c4)}\\
\end{aligned}
\end{equation}

\noindent where contraints (c1)$\sim$(c3) are exactly the same constraints as that in the original problem formualtion.

For the above optimization problem, we have the following routing structure for its optimal solution.

\begin{theorem}
\label{thrm:no_splitting_bandwidth}
The optimal solution to problem~\eqref{eq:mulcache_model_bandwidth} is the one such that each CP directs its requests to at most one cache at the volume less than the maximum volume, and it either does not direct or directs requests at the maximum volume to the other caches.
\end{theorem}

\begin{proof}
From Appendix~\ref{sec:proof_thm1}, we know that by moving requests of each CP from one cache to another which allocates a larger cache slice, the hit rate can be increased. Furthermore, since there is bandwidth limitation between each CP and each cache, when the bandwidth limitation to a cache is reached, each CP will try to route its remaining requests to another cache which allocates to it the next largest cache slice so as to maxmize its hit rate, until the bandwidth limitation is reached. This process goes on and terminates when all requests of the CP are routed, with a routing configuration given by the theorem.
\end{proof}

Similarly, with small problem sizes, we can devise efficient algorithm to solve problem~\eqref{eq:mulcache_model_bandwidth} by converting it into a series of subproblems $Pr_1,Pr_2,...,Pr_L$ where each $Pr_i$ corresponds to a cache resource
allocation problem with one specific request routing, i.e., the request routing of each CP is as such that described by Theorem~\ref{thrm:no_splitting_bandwidth}. Note that each $Pr_i$ is convex as we have proved in Appendix~\ref{sec:proof_thm2}. The optimal solution to problem~\eqref{eq:mulcache_model_bandwidth} is the one with the largest objective function value.

Likewise, we can prove the problem has a biconvexity structure and hence obtain its efficient solution algorithms.
However, not all solutions corresponding to the routing configurations given by Theorem~\ref{thrm:no_splitting_bandwidth} are partial optimums.
Instead, they can be obtained using alogrithm {\it{ACS (Alternate Convex Search)}}~\cite{wendell1976minimization} by first optimizing cache allocation variables, with starting points whose routing configurations are given by Theorem~\ref{thrm:no_splitting_bandwidth}. 


\subsubsection{Decentralized Mechanism}
Based on Theorem~\ref{thrm:no_splitting_bandwidth} and using the same Lagrangian-based dual decomposition technique as in Section~\ref{sec:decentralized_algorithm}, we can formulate CP $k$'s problem with a given request routing $\mathrm{\textbf{P}_k}$'s as

\noindent \textbf{CP $k$'s problem:}

\begin{equation} \label{eq:CP_problem_bandwidth}
\begin{aligned}
  \mathop \text{maximize }\limits_{\mathrm{\textbf{C}_k}} (w_kU_{k}(h_k(\mathrm{\textbf{C}_k}))-\sum_{m \in H(k)}\lambda_mC_{km})
\end{aligned}
\end{equation}

\noindent where

\begin{itemize}
\item $U_{k}(h_k(\mathrm{\textbf{C}_k}))=U_k( \sum_{m \in I(k)}h_{km}(C_{km}, \frac{V_{km}}{\sum_{i=1}^{N_k}\lambda_{ik}})+h_{ks}(C_{ks}, 1-\frac{\sum_{m \in I(k)}V_{km}}{\sum_{i=1}^{N_k}\lambda_{ik}}))$;
\item $I(k)$ is defined as the set of caches to which CP $k$ directs its requests at the maximum volume;
\item $s$ is the cache to which CP $k$ directs its residual requests. 
\end{itemize}

\noindent Note that ($I(k)$, $s$) is determined by $\mathrm{\textbf{P}_k}$.

Using the same reasoning as in the proof of Theorem~\ref{thrm:concavity_CP_problem}, we can prove the following theorem.

\begin{theorem}
\label{thrm:concavity_2}
The cache resource allocation problem~\eqref{eq:CP_problem_bandwidth} has a unique optimal solution.
\end{theorem}

The size of cache slice that CP $k$ requires from cache $m$ is thus:

\begin{equation} \label{eq:CP_cache_silice_bandwidth}
\begin{aligned}
  \text{arg  }\{\mathop \text{maximize } \limits_{C_{km}, m \in I(k)\cup \{s\} } w_kU_{k}(h_k(\mathrm{\textbf{C}_k)}-\sum_{m \in I(k)\cup \{s\}}\lambda_mC_{km}\}
\end{aligned}
\end{equation}

\noindent Note that $C_{km}=0$ if $m \notin I(k)\cup \{s\}$ as CP $k$ will not require cache resources from these caches. The algorithm for each cache to update its price remains unchanged as that given by Eq.~\eqref{eq:Cache_adaptation}.


\subsection{Latency Optimization}
This subsection considers user-centric problem formulation. In particular, we consider the scenario where the objective is to optimize user-perceived content delivery latency.

\subsubsection{Problem Formulation}
Let $d_{km}$ be the network delay of fetching content of CP $k$ from cache $m$ if the requested
file is cached, and $d^{0}_{km}$ be the delay of fetching it from content custodians (remote servers), i.e., when the request is missed and cache $m$ forwards it to the back-end server, downloads the file and forwards it back to the user. For each CP $k$, denote by $t_{km}$ its average latency of fetching content through cache $m$, and $t_k$ the overall average content delivery latency. We have

\[t_{km}(C_{km}, p_{km})=\frac{d_{km}h_{km}+d^0_{km}(\sum_{i=1}^{N_k}\lambda_{ik}p_{km}-h_{km})}{\sum_{i=1}^{N_k}\lambda_{ik}p_{km}}\]

\noindent and

\[t_k=\sum_{m=1}^{M}t_{km}p_{km}=\sum_{m=1}^{M}\frac{d_{km}h_{km}+d^0_{km}(\sum_{i=1}^{N_k}\lambda_{ik}p_{km}-h_{km})}{\sum_{i=1}^{N_k}\lambda_{ik}}\]

Furthermore, associated with each CP $k$ is a utility $U_k(t_k)$ that is a concave, decreasing function of $t_k$. We have the latency-optimization formulation as follows:

\[ \text{maximize } \sum_{k=1}^K w_kU_{k}(t_{k}(\mathrm{\textbf{C}_k}, \mathrm{\textbf{P}_k})) \]

\noindent subject to

\begin{equation} \label{eq:latency_model}
\begin{aligned}
  &\sum_{k=1}^K C_{km} \leq C_m,  \quad  \forall m \\
  &\sum_{m=1}^M p_{km} = 1, \quad \forall k \\
  &0 \leq p_{km} \leq a_{km}, \quad \forall k, m \\
\end{aligned}
\end{equation}

Since $t_{km}$ is linear in $h_{km}$, and $U_k(t_k)$ is a concave decreasing function of $t_k$, we can prove that problem~\eqref{eq:latency_model} has similar properties as problem~\eqref{eq:multi_cache_model}, i.e., Theorem~\ref{thrm:no_splitting}, Theorem~\ref{thrm:concavity} and Theorem~\ref{thrm:biconvexity} hold for problem~\eqref{eq:latency_model}. Therefore, the same centralized algorithm with the same complexity can be applied to obtain solutions.

\subsubsection{Decentralized Algorithm}

Using Lagrangian-based dual decomposition technique as in Section~\ref{sec:decentralized_algorithm}, we can formulate CP $k$'s problem with a given request routing $\mathrm{\textbf{P}_k}$'s as:

\noindent \textbf{CP $k$'s problem:}

\begin{equation} \label{eq:CP_problem_latency}
\begin{aligned}
  \mathop \text{maximize }\limits_{\mathrm{\textbf{C}_k}} (w_kU_{k}(t_k(\mathrm{\textbf{C}_k}))-\lambda_{s(k)}C_{ks(k)})
\end{aligned}
\end{equation}

\noindent where we define $U_{k}(t_k(\mathrm{\textbf{C}_k}))=U_k(t_{ks(k)}(C_{ks(k)}, 1))$, i.e., CP $k$ routes all its requests to cache $s(k)$ and the required cache amount is such that it leads to the smallest content delivery latency.

It can be proved that problem~\eqref{eq:CP_problem_latency} is convex and as a result, the same decentralized algorithm as in Section~\ref{sec:decentralized_algorithm} can be applied to obtain the optimal request routing and cache allocations.


%



\section{Numerical Studies and Evaluation}
\label{sec:evaluation}
In this section, we present numerical results to show: 1) the efficacy of our joint cache resource allocation and request routing mechanism and 2) the convergence of our decentralized algorithms to optimal solution.

\subsection{Evaluation Setup}

We consider a cache network which comprises of 3 caches and 2 content providers. CP 1 connects to Cache 1 and Cache 2, while CP 2 connects to Cache 2 and Cache 3. Thus Cache 2 is shared by both content providers. The three caches adopt LRU policies and are of size $C_1=500$, $C_2=1200$, and $C_3=500$, respectively. The two content providers serve $N_1 = 10^4$ and $N_2 = 2\times 10^4$ content files. File popularities for the two providers follow Zipf distributions with parameters $\alpha_1 = 0.6$ and $\alpha_2 = 0.8$, respectively. Requests for the files from the two content providers arrive as Poisson processes with aggregate rates $R_1 = 10$ and $R_2=15$. We set $w_1 = w_2 = 1$ assuming that the two content providers are equally important to the cache manager. The utilities of the two providers are expressed as $U_1(h) = U_2(h)=h$ so that the goal is to maximize the cache utilization efficiency (aggregate cache hit rates).

\subsection{Results for Basic Model}

\subsubsection{Efficacy of the mechanism}
In order to validate its efficacy, we compare the overall utility observed by our mechanism to that by simple static routing, i.e., when both content providers route their requests to the connected caches at equal probabilities. Note that even under simple static routing, cache resource allocation is obtained by solving the corresponding convex optimization problem. Figure~\ref{fig:utility_comparison} shows how the utility changes as the capacity of Cache 2 varies. It can be seen that under all cache capacities our mechanism outperforms. Moreover, the larger the shared cache, the more utility gain.

\begin{figure}
  \centering
  \includegraphics[width=0.28\textwidth]{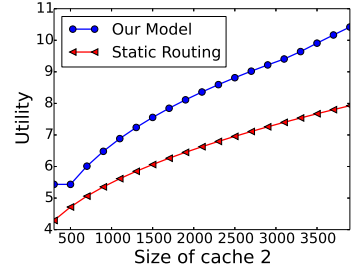}\\
  \caption{Efficacy of our model as compared to static routing.}\label{fig:utility_comparison}
\end{figure}

Figure~\ref{fig:traffic_distribution} shows traffic distributions of the two content providers, from which we observe obvious hand-overs. More specifically, it can be seen that when the size of shared cache is small ($300 \le C_2 \le 500$), both content providers use its dedicated cache only; when the shared cache becomes larger ($500 < C_2 \le 3100$), CP 2 routes all its traffic to Cache 2; and when it continues to grow ($C_2 > 3100$), both providers route all their traffic to the shared cache.

\begin{figure}
    \centering
    \subfigure[CP 1]
    {
        \includegraphics[width=0.225\textwidth]{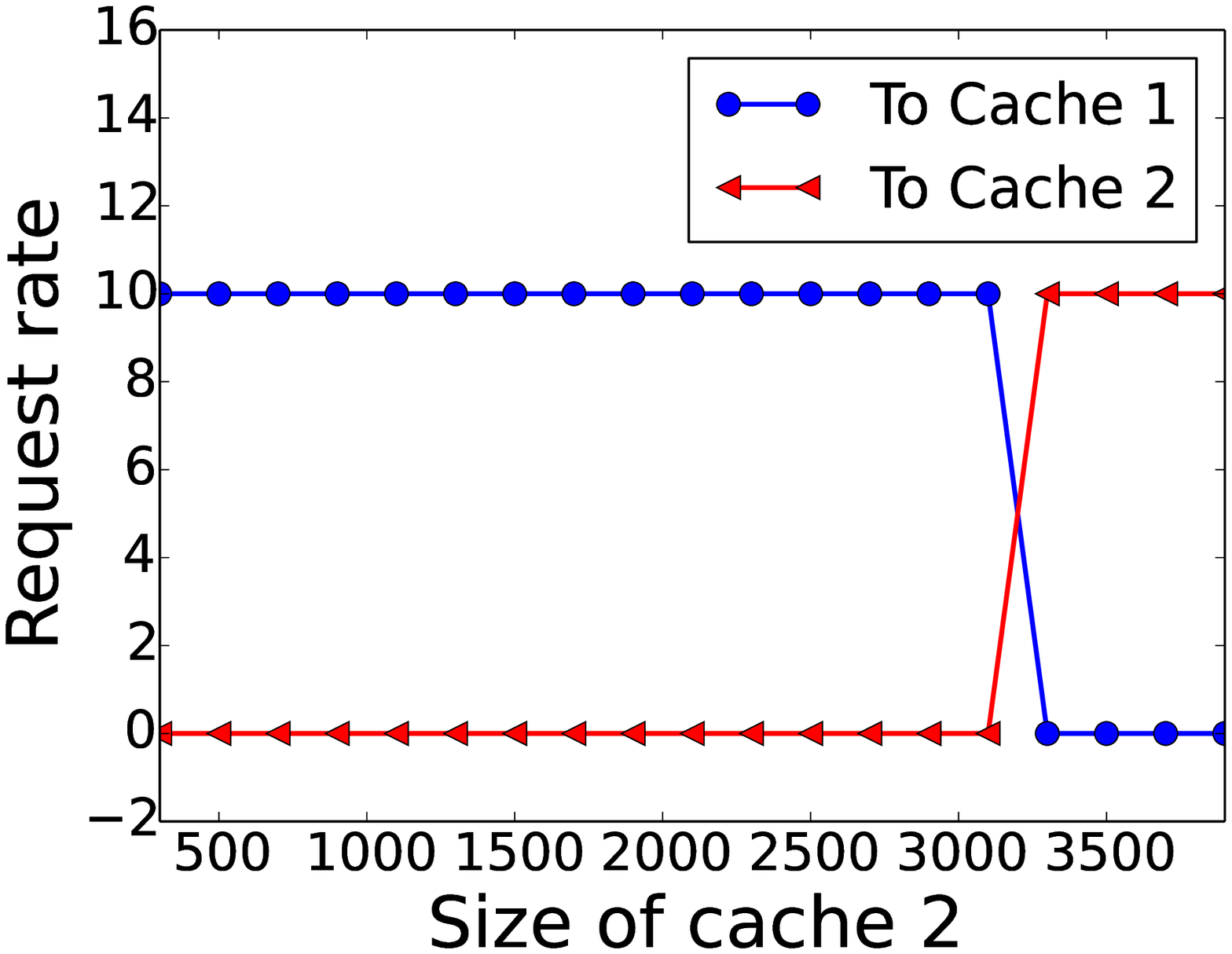}
        \label{fig:first_sub}
    }
    \subfigure[CP 2]
    {
        \includegraphics[width=0.225\textwidth]{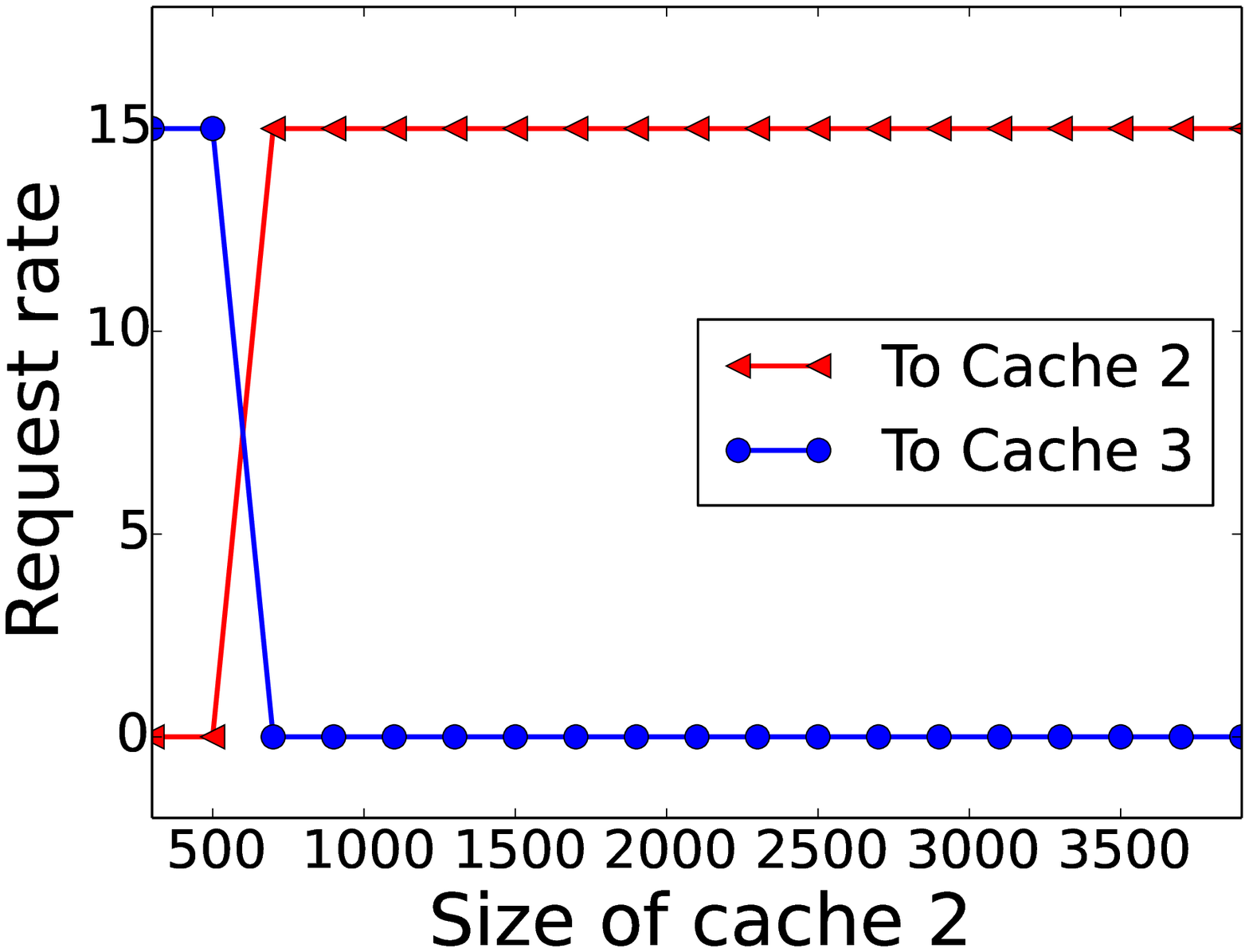}
        \label{fig:second_sub}
    }
    \caption{Traffic distribution of the two content providers.}
    \label{fig:traffic_distribution}
\end{figure}

Figure~\ref{fig:partition} shows how the shared source is allocated to the two providers.
As expected, it can be seen that all the cache resource is allocated to CP 2 when only CP 2 directs requests to it. The cache resource is shared when both providers route their traffic to the cache. Note that the amount of resource allocated to them is determined by the optimization model.

\begin{figure}
  \centering
  \includegraphics[width=0.28\textwidth]{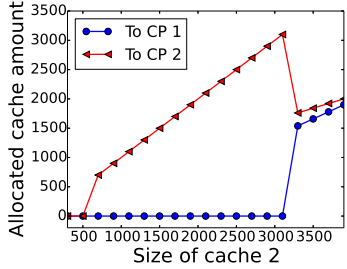}\\
  \caption{Resource allocation of Cache 2 to the two content providers.}\label{fig:partition}
\end{figure}

From Figure~\ref{fig:hit_rate} we observe that the hit rate of a content provider will not keep monotonically increasing as the shared cache resource increases. In particular, we observe that the hit rate of CP 2 decreases when the capacity of Cache 2 varies from 3000 to 3100. This phenomena is due to the fact that when the size of Cache 2 is less than 3100, only CP 2 routes its traffic to Cache 2 and hence the resource is solely allocated to it; when the size of Cache 2 reaches 3100, CP 1 begins to route its traffic to it and hence the resource is shared among both providers, which certainly results in a decrease of the hit rate for CP 2. Nevertheless, the overall hit rate keeps increasing, as expected.

\begin{figure}
  \centering
  \includegraphics[width=0.28\textwidth]{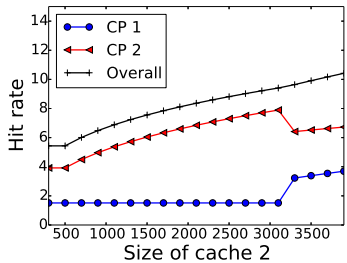}\\
  \caption{Observed hit rates as a function of the shared cache size.}\label{fig:hit_rate}
\end{figure}

\subsubsection{Impact of different parameters}
We next look at the effect of various parameters on the competing process for the shared cache resources by the two content providers. To achieve this, we set $C_2 = 5000$ (large enough) so that both providers will have the opportunity to direct their requests to Cache 2. We fix the parameter of CP 2, and study the effect of changing the aggregate request rate $R_1$ of CP 1, the weight parameter $w_1$, and the skewness parameter of the Zipfian file popularity distribution. Fig.~\ref{fig:parameter_impact} shows the effect of these different parameters on the observed hit rates and allocated cache amount of Cache 2 to the two providers.

\noindent (1) {\bf{Request rate:}} as $R_1$ increases, more cache resource is allocated to CP 1, until Cache 2 is solely occupied by CP 1. Consequently, the hit rate of CP 1 increases. Note that the hit rate of CP 1 grows exactly in linear with its request rate when $R_1 > 20$ since from then on the allocated cache resource to CP 1 does not change anymore. The results indicate that if the goal is to maximize the overall hit rate, then content providers with a larger request rate will have priority over the others in the competing process. 

\noindent (2) {\bf{Skewness parameter:}} when $\alpha_1$ increases, less cache resource is allocated to CP1. This is because with a large $\alpha_1$, a small fraction of files generates most of the traffic. Also it is surprisingly to observe that the hit rate of CP 1 increases and so does the overall hit rate. From this point of view, we conclude that a larger skewness of traffic distribution not only benefit content providers generating the traffic, but also others in system.

\noindent (3) {\bf{Weight:}} Clearly, the weight of content providers significantly influences the competing process for the shared cache resources. It is observed that: i) the hit rate of CP 1 as well as cache resource allocated to it grows with the increase of its weight $w_1$, until CP 1 fully occupies Cache 2; ii) although the overall hit rate also grows as we show in the figure, we argue that this is not the case in general. In fact, since there are only two content providers, the figures can also be interpreted as if we change the weight of CP 2 while fixing that of CP 1. Then the hit rate of the two content providers as well as that of the system will go into the reverse direction. Indeed, as the objective function of problem~\eqref{eq:multi_cache_model} shows, the overall hit rate not only depends on the weights of content providers, but also depends on their traffic characteristics.

\begin{figure*}[t]
    \centering
    \subfigure
    {
        \includegraphics[width=0.30\textwidth]{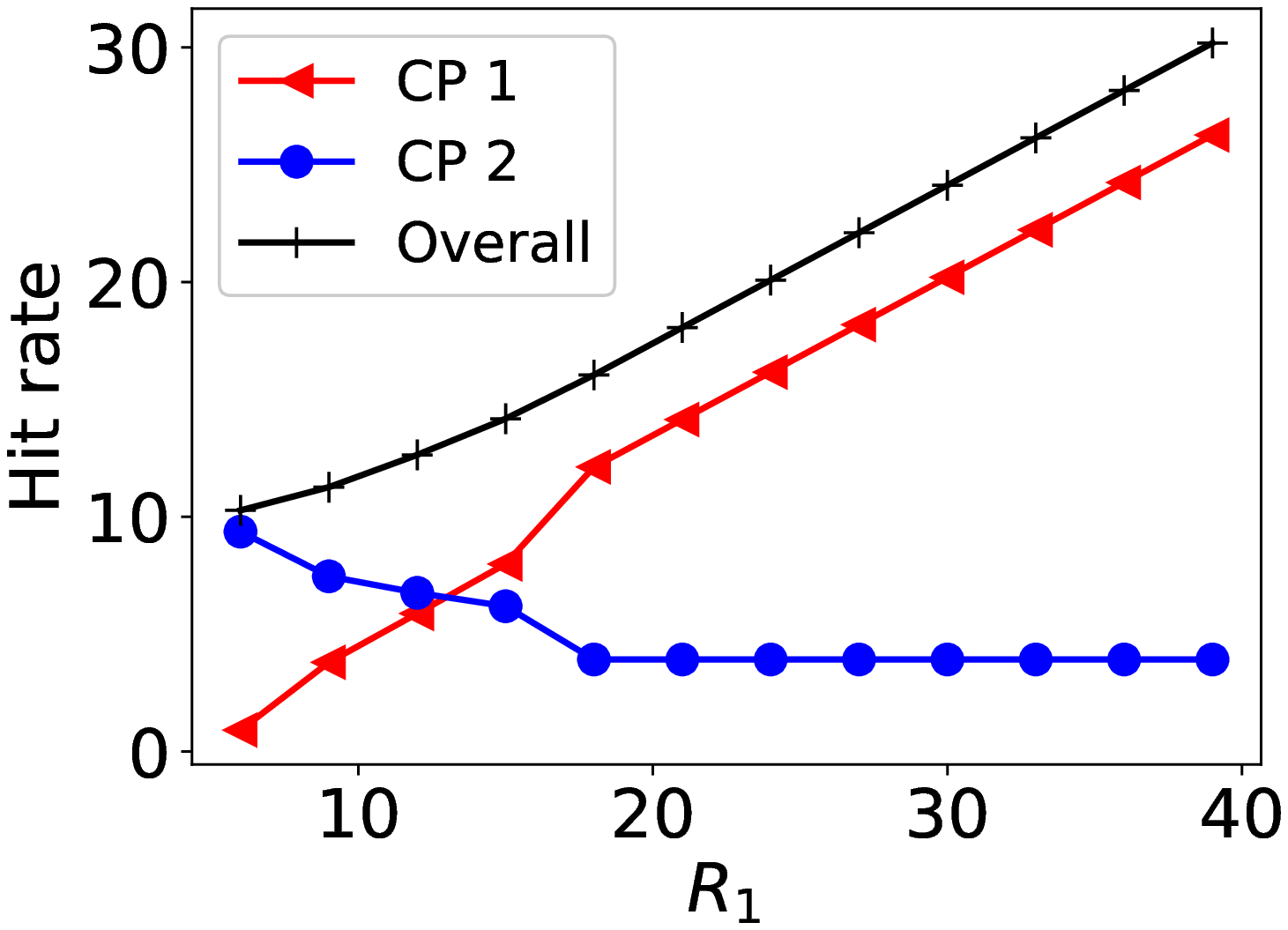}
    }
    \subfigure
    {
        \includegraphics[width=0.30\textwidth]{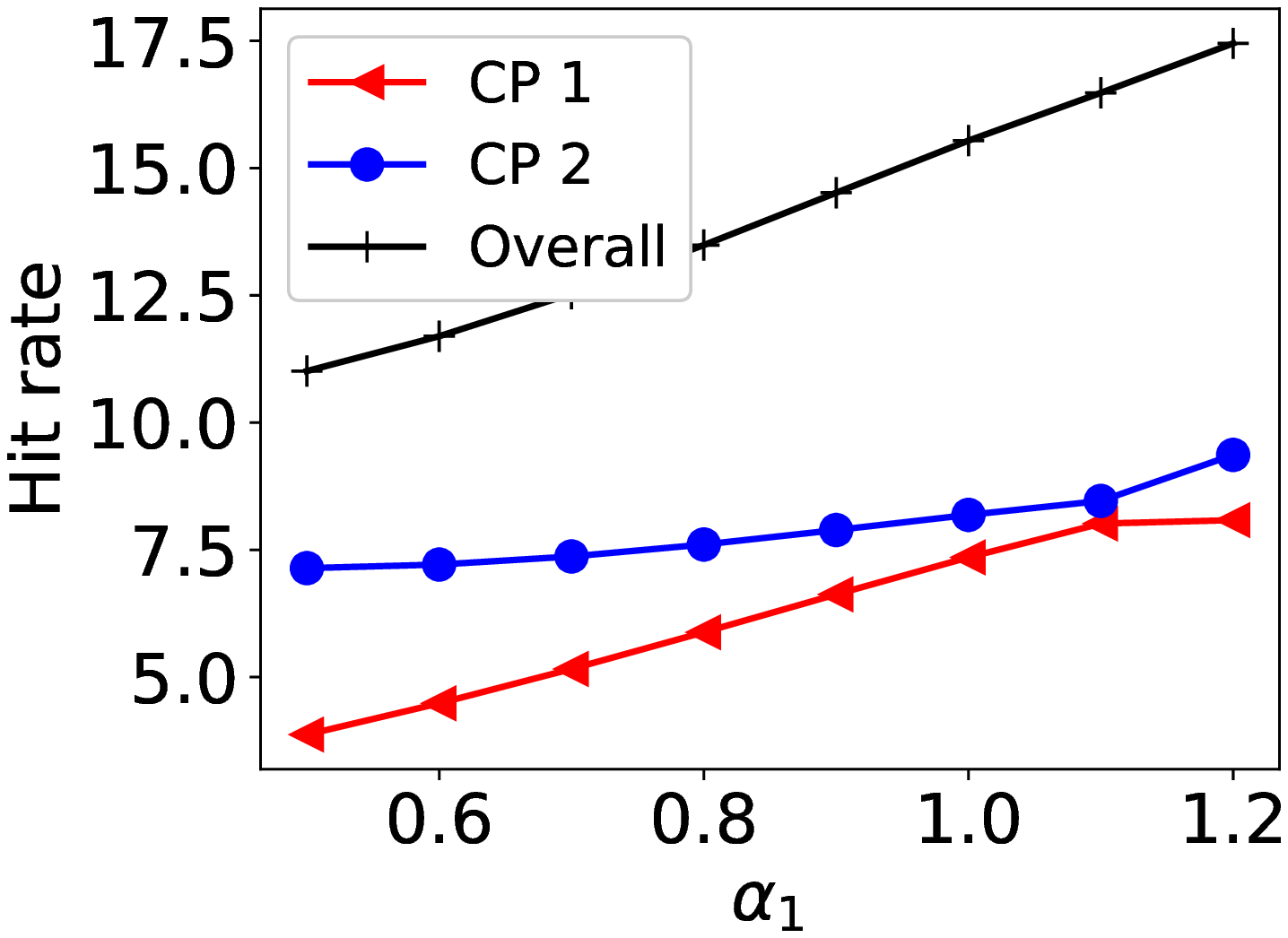}
    }
    \subfigure
    {
        \includegraphics[width=0.30\textwidth]{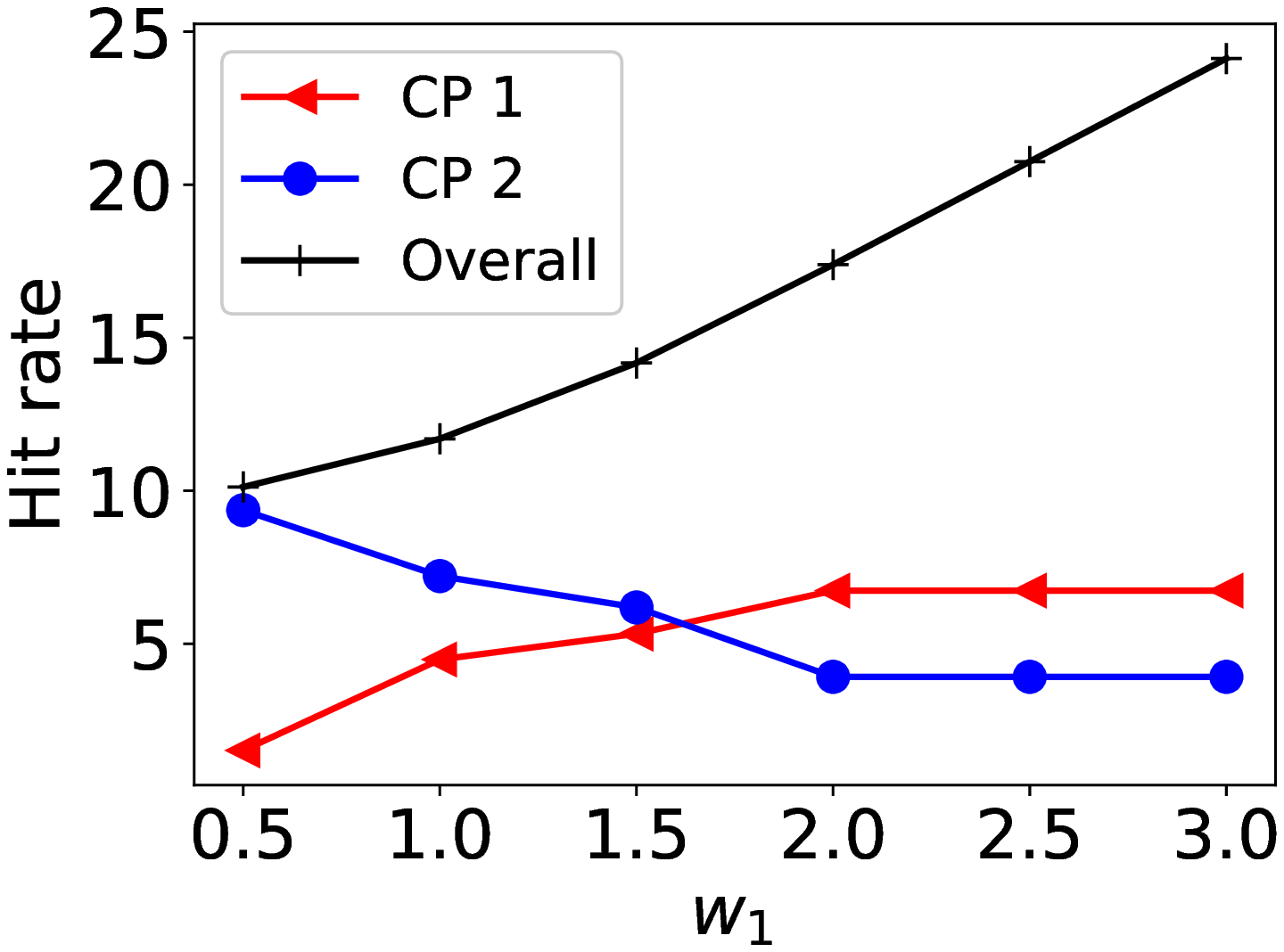}
    }\\
    \subfigure
    {
        \includegraphics[width=0.30\textwidth]{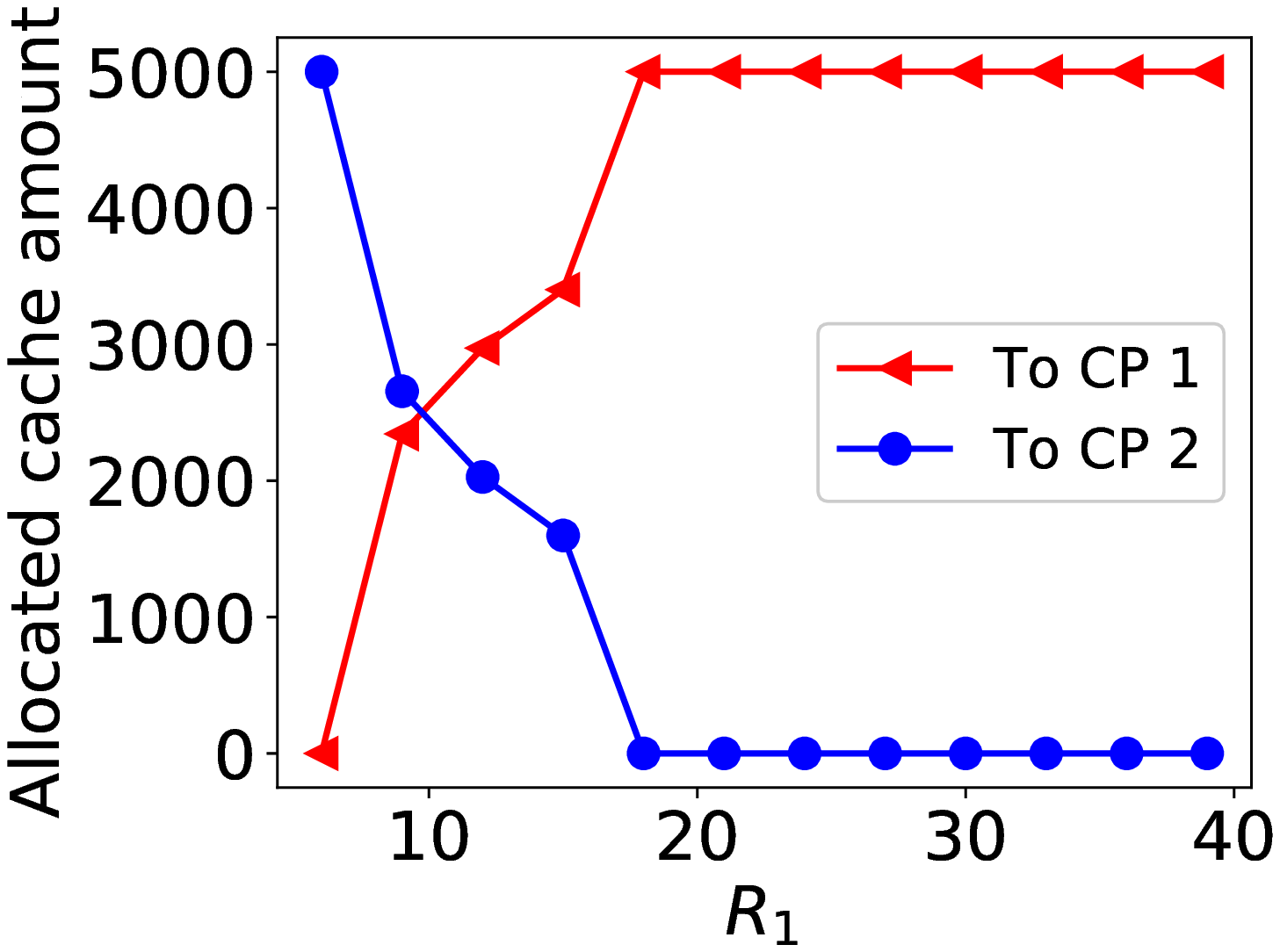}
    }
        \subfigure
    {
        \includegraphics[width=0.30\textwidth]{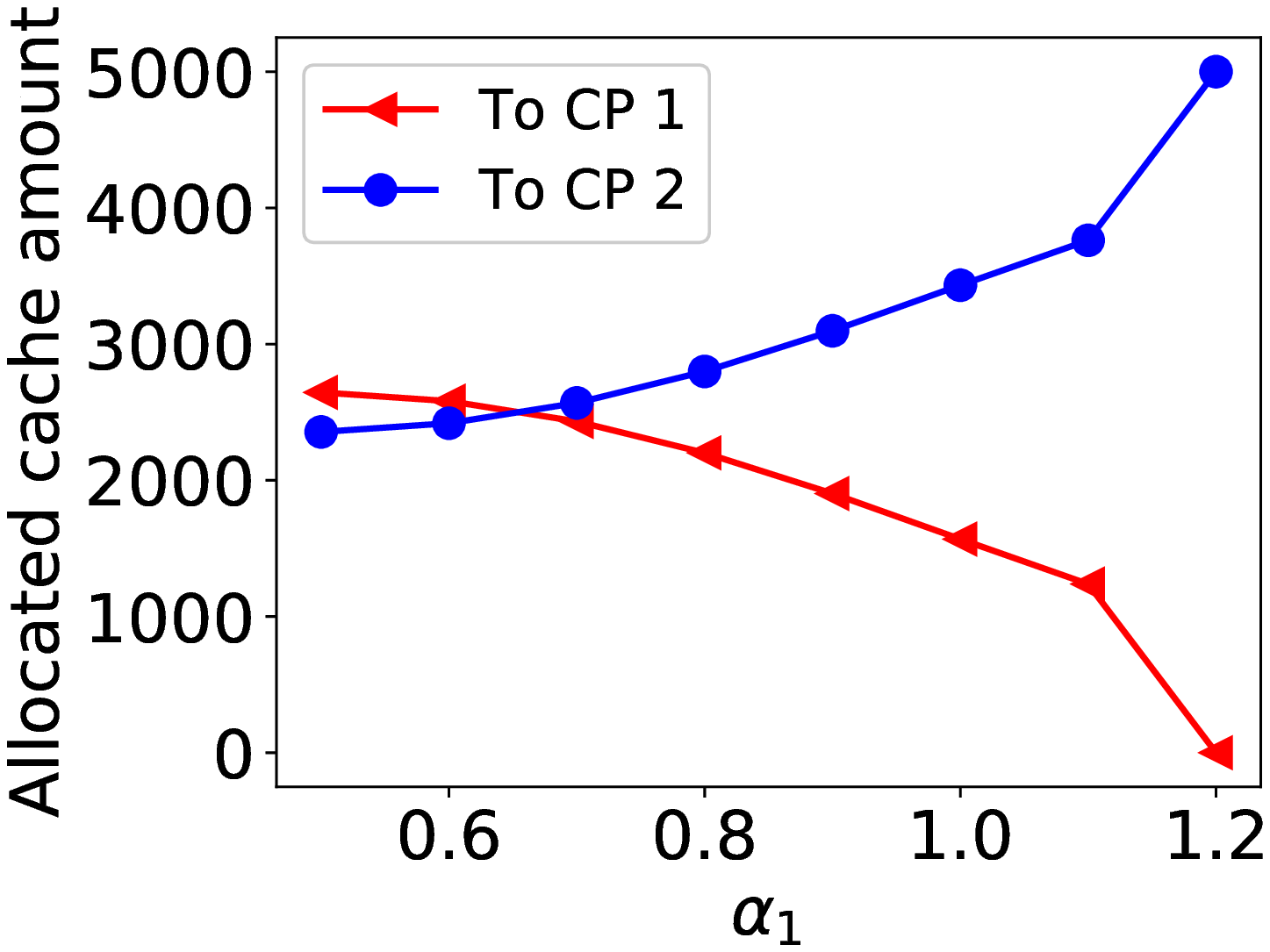}
    }
    \subfigure
    {
        \includegraphics[width=0.30\textwidth]{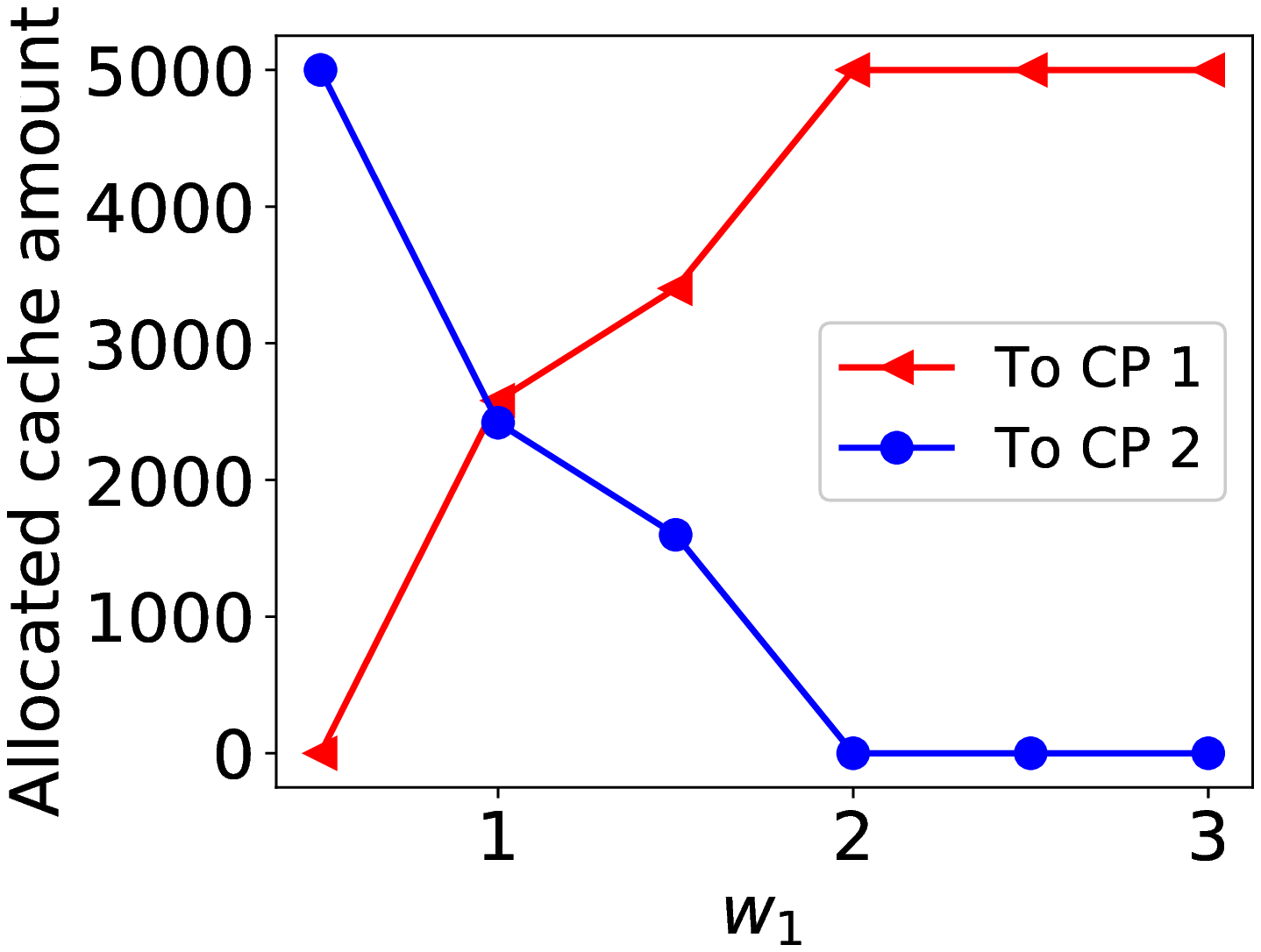}
    }
    \caption{Effect of the parameters on hit rates and cache resources allocated to content providers.}
    \label{fig:parameter_impact}
\end{figure*}

\subsubsection{Convergence of decentralized algorithm}
To investigate the performance of our decentralized algorithm, we fix the size of Cache 2 as $C_2=1900$ and choose the initial prices for the three caches as $\lambda_1^0=\lambda_2^0=\lambda_3^0=0$. The step size is set as $r=10^{-6}$. Figure~\ref{fig:dynamics} shows system dynamics as time goes on (here ``CP $k$ - Cache $m$'' means CP $k$ routes all its traffic to Cache $m$). It can be seen that our algorithm converges under all four request routings, and the optimal solution is the one with the request routing ``CP 1 - Cache 1, CP 2 - Cache 2", that is, CP 1 directs all its requests to Cache 1 and CP 2 directs all its requests to Cache 2. Looking back at Figure~\ref{fig:traffic_distribution} and Figure~\ref{fig:partition}, we can see that this result is in accordance with the centralized solution. Therefore, our algorithm converges to the optimal solution.

\begin{figure*}[t]
    \centering
    \subfigure
    {
        \includegraphics[width=0.15\textwidth]{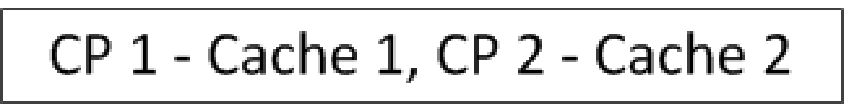}
    }\\
    \subfigure
    {
        \includegraphics[width=0.22\textwidth]{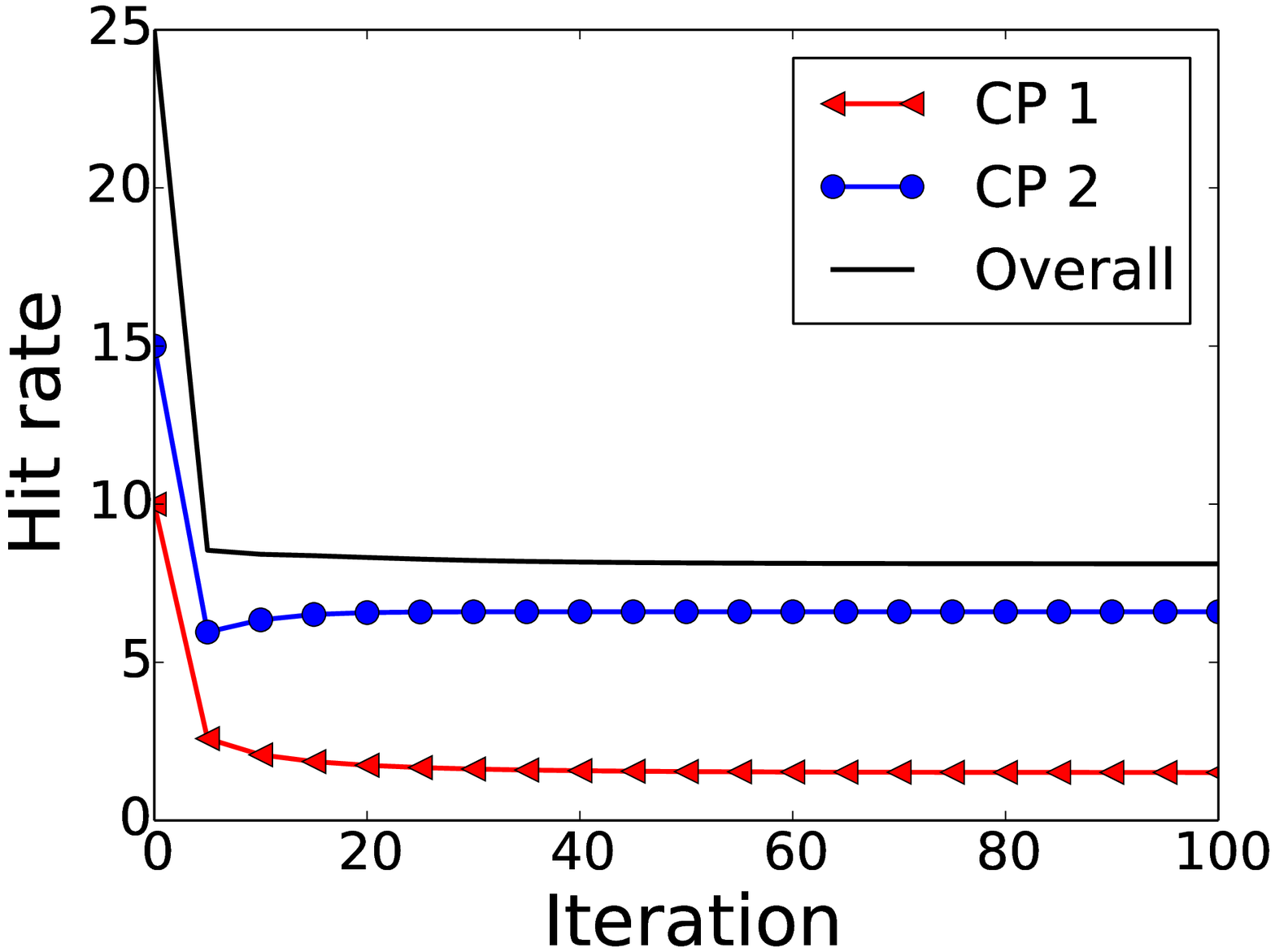}
    }
    \subfigure
    {
        \includegraphics[width=0.22\textwidth]{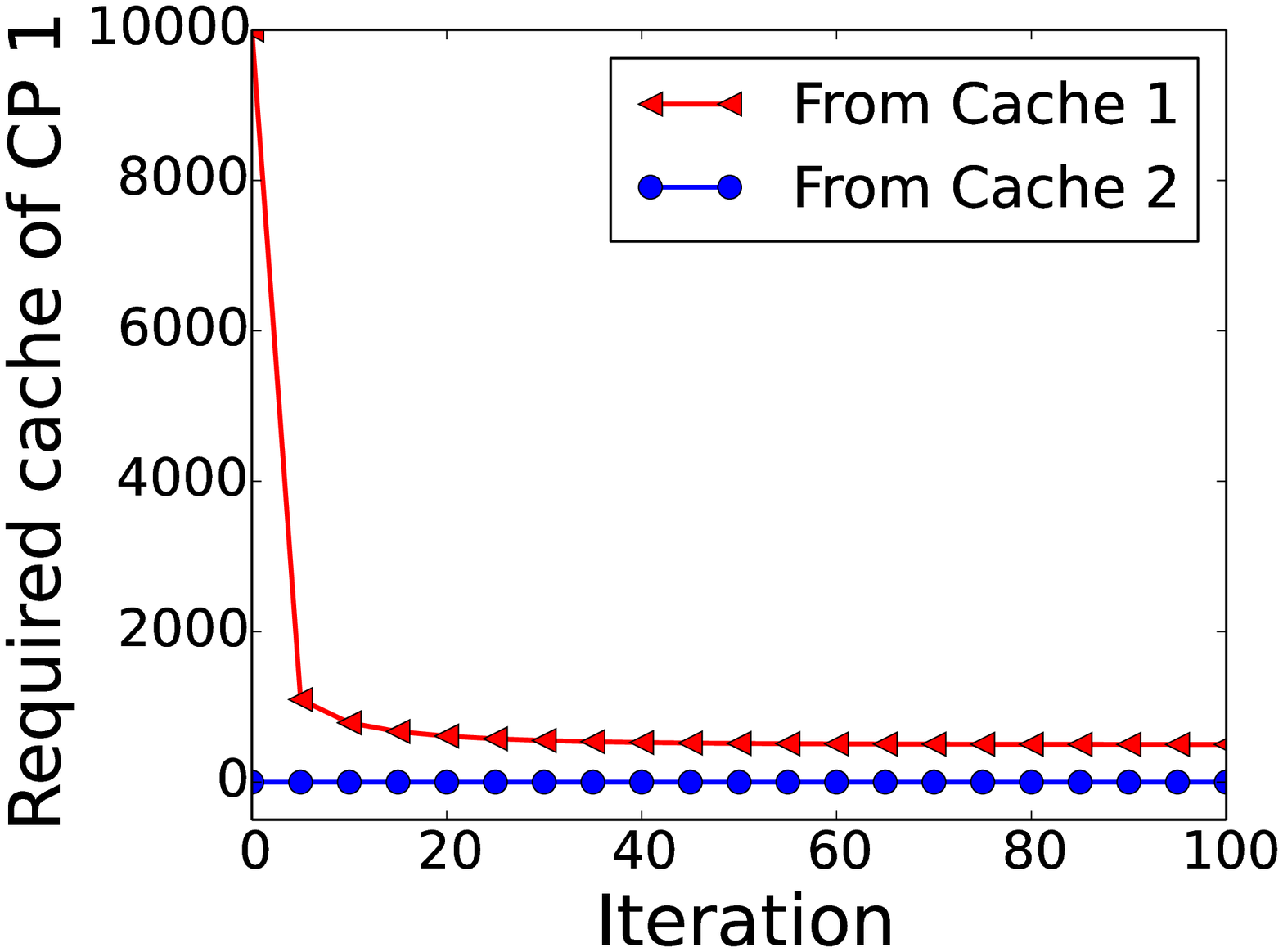}
    }
    \subfigure
    {
        \includegraphics[width=0.22\textwidth]{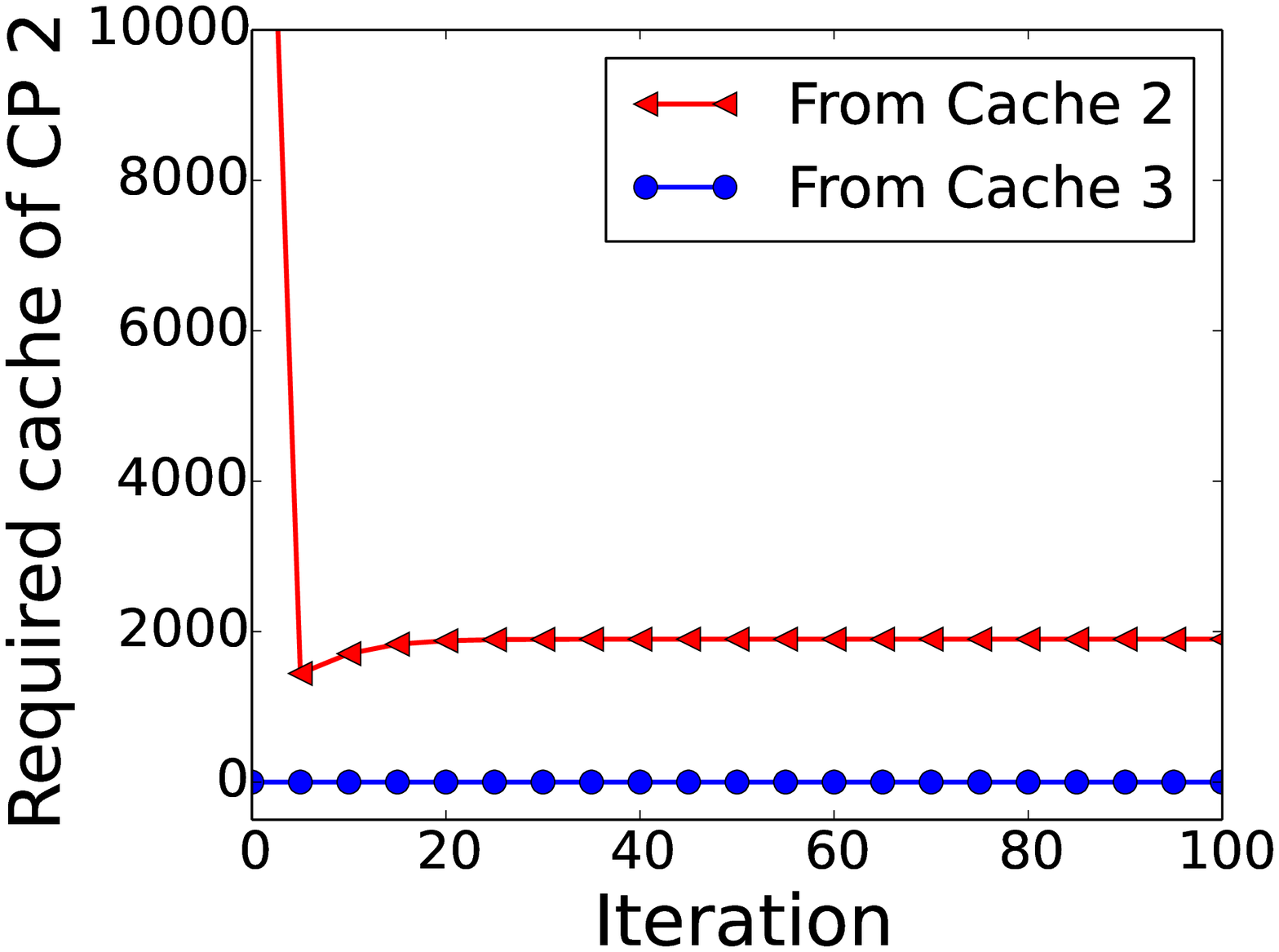}
    }
        \subfigure
    {
        \includegraphics[width=0.22\textwidth]{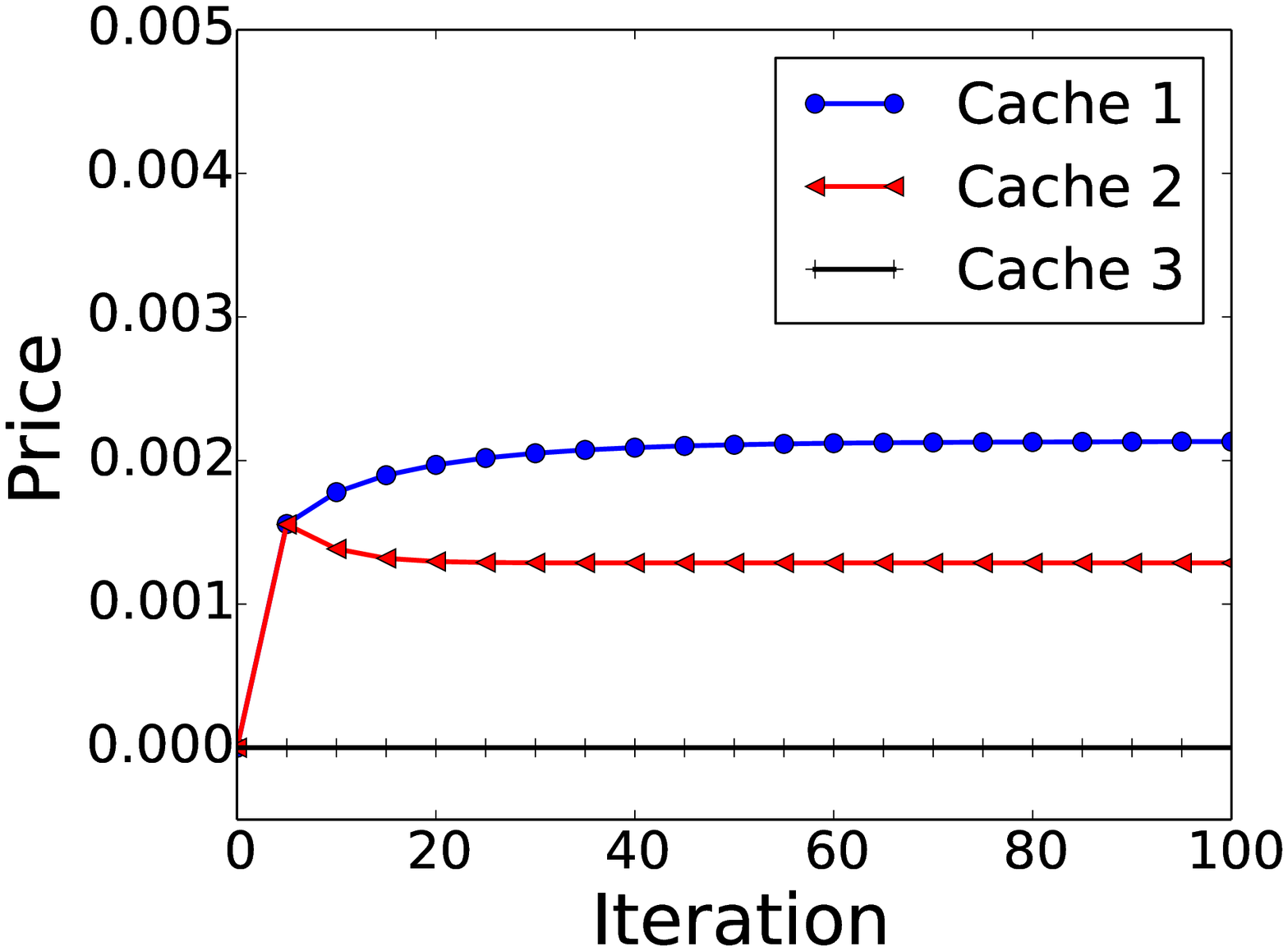}
    }\\
    \subfigure
    {
        \includegraphics[width=0.15\textwidth]{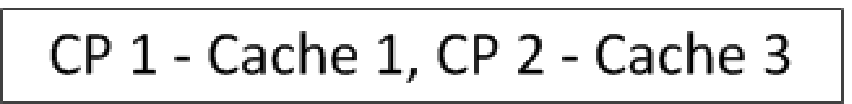}
    }\\
    \subfigure
    {
        \includegraphics[width=0.22\textwidth]{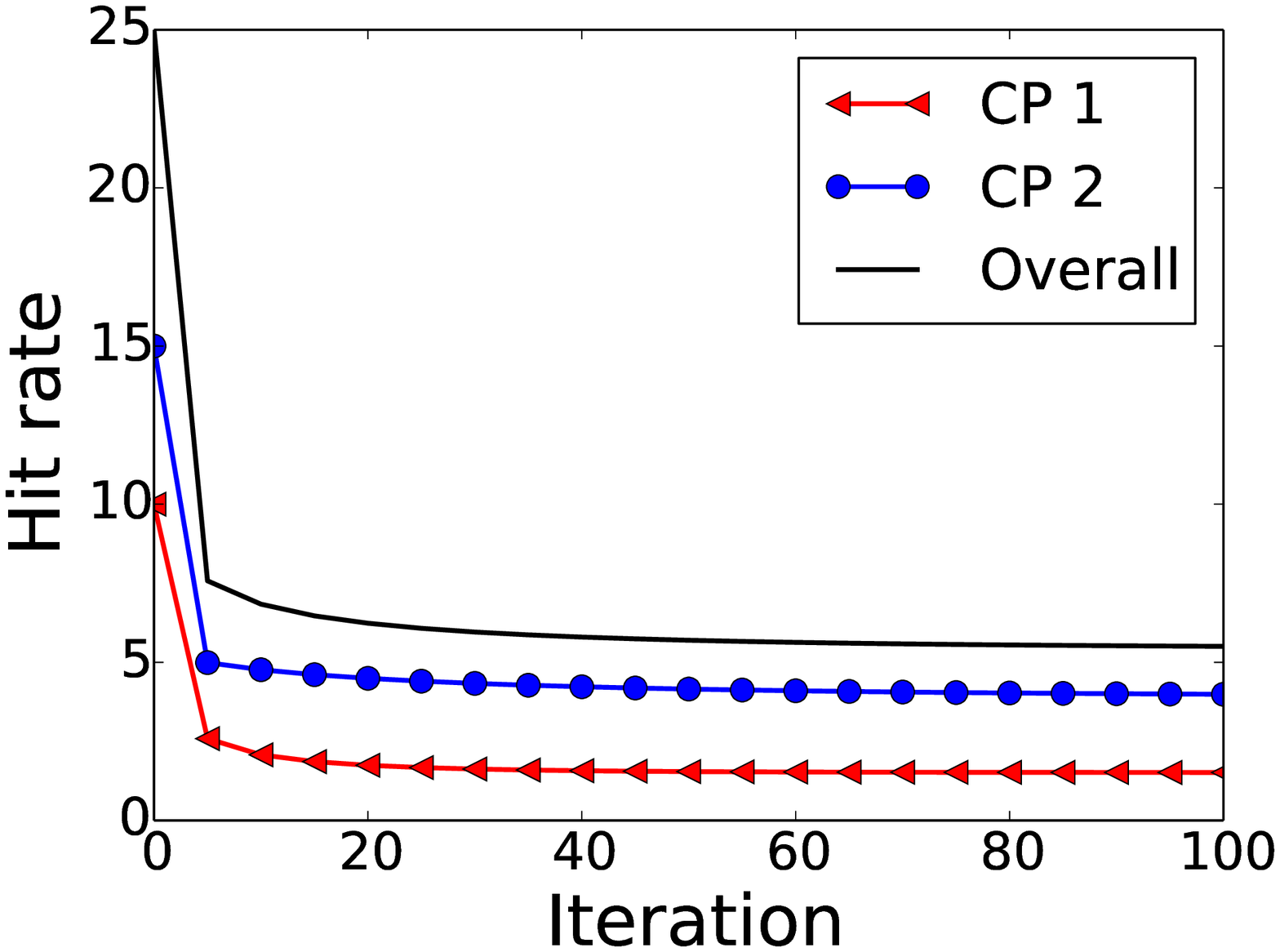}
    }
    \subfigure
    {
        \includegraphics[width=0.22\textwidth]{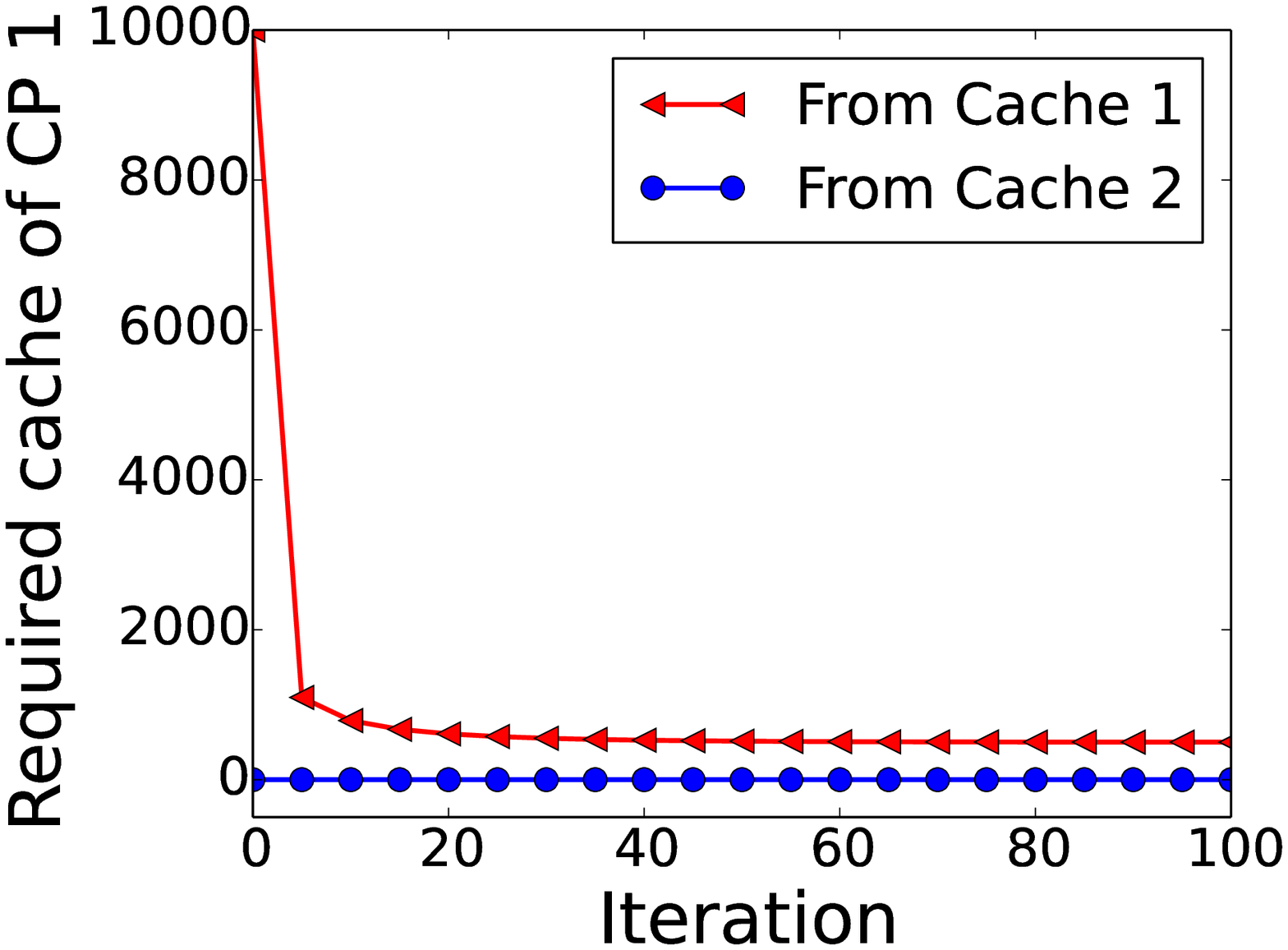}
    }
    \subfigure
    {
        \includegraphics[width=0.22\textwidth]{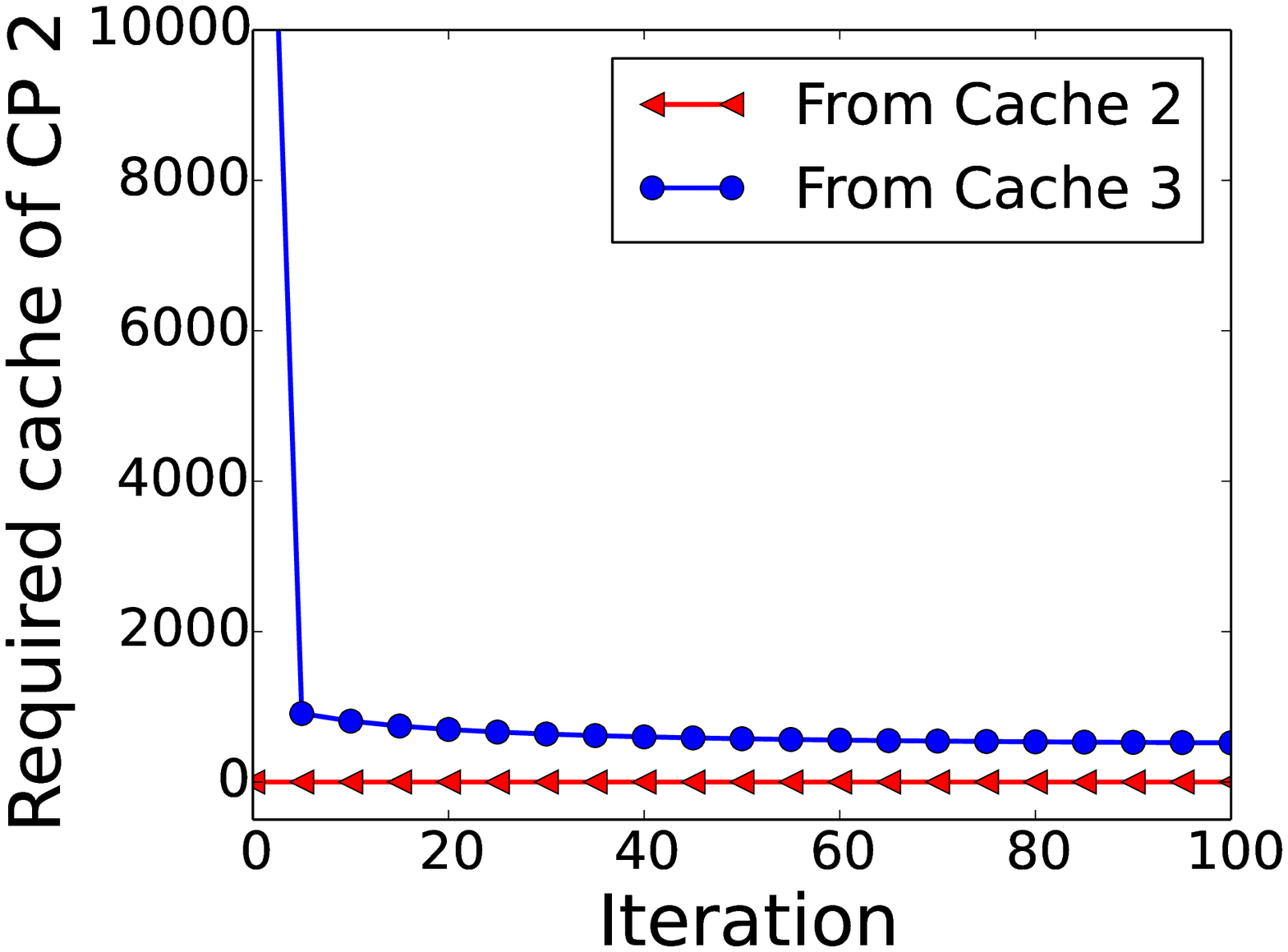}
    }
        \subfigure
    {
        \includegraphics[width=0.22\textwidth]{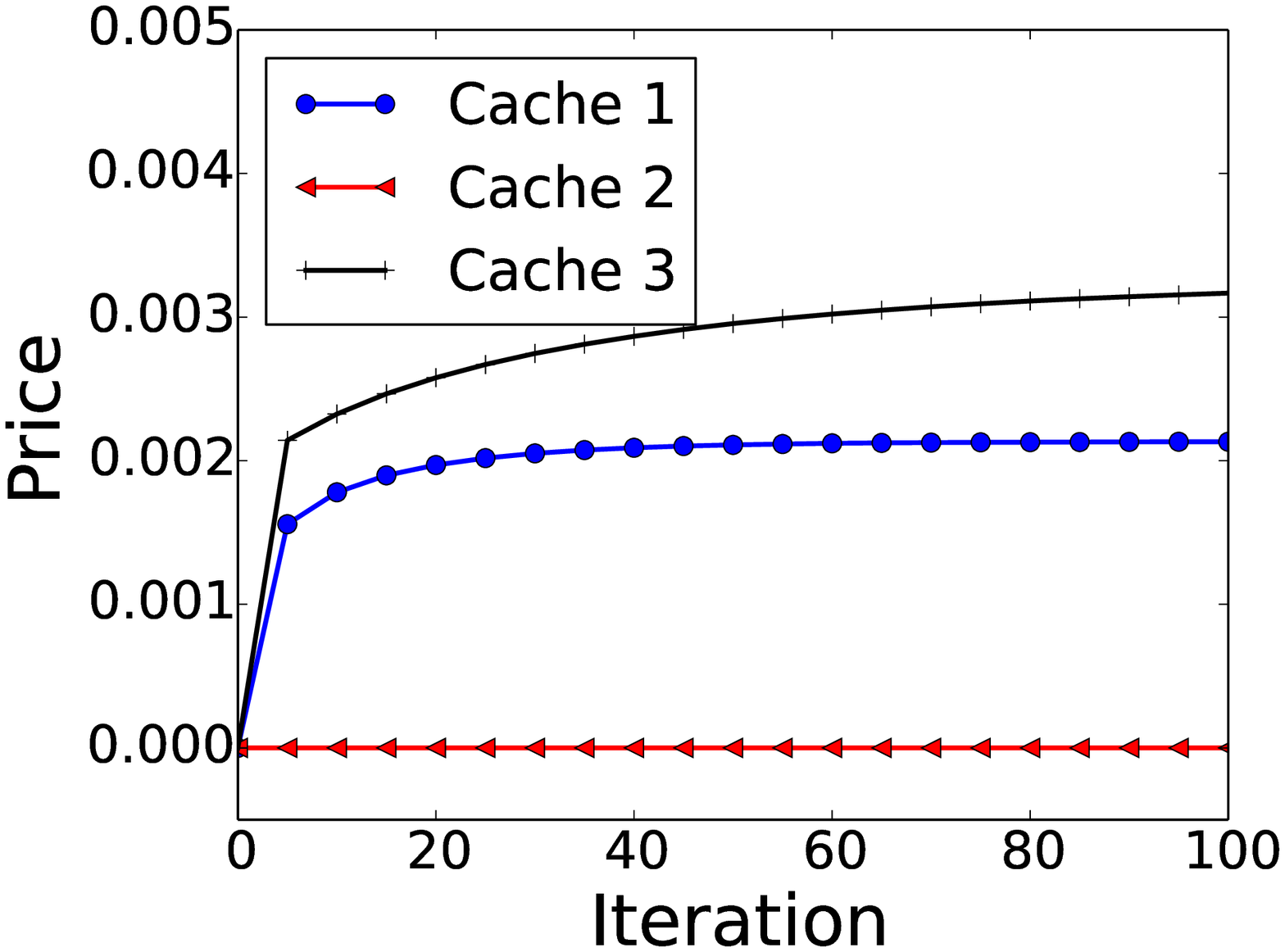}
    }\\
    \subfigure
    {
        \includegraphics[width=0.15\textwidth]{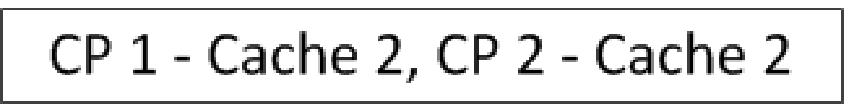}
    }\\
    \subfigure
    {
        \includegraphics[width=0.22\textwidth]{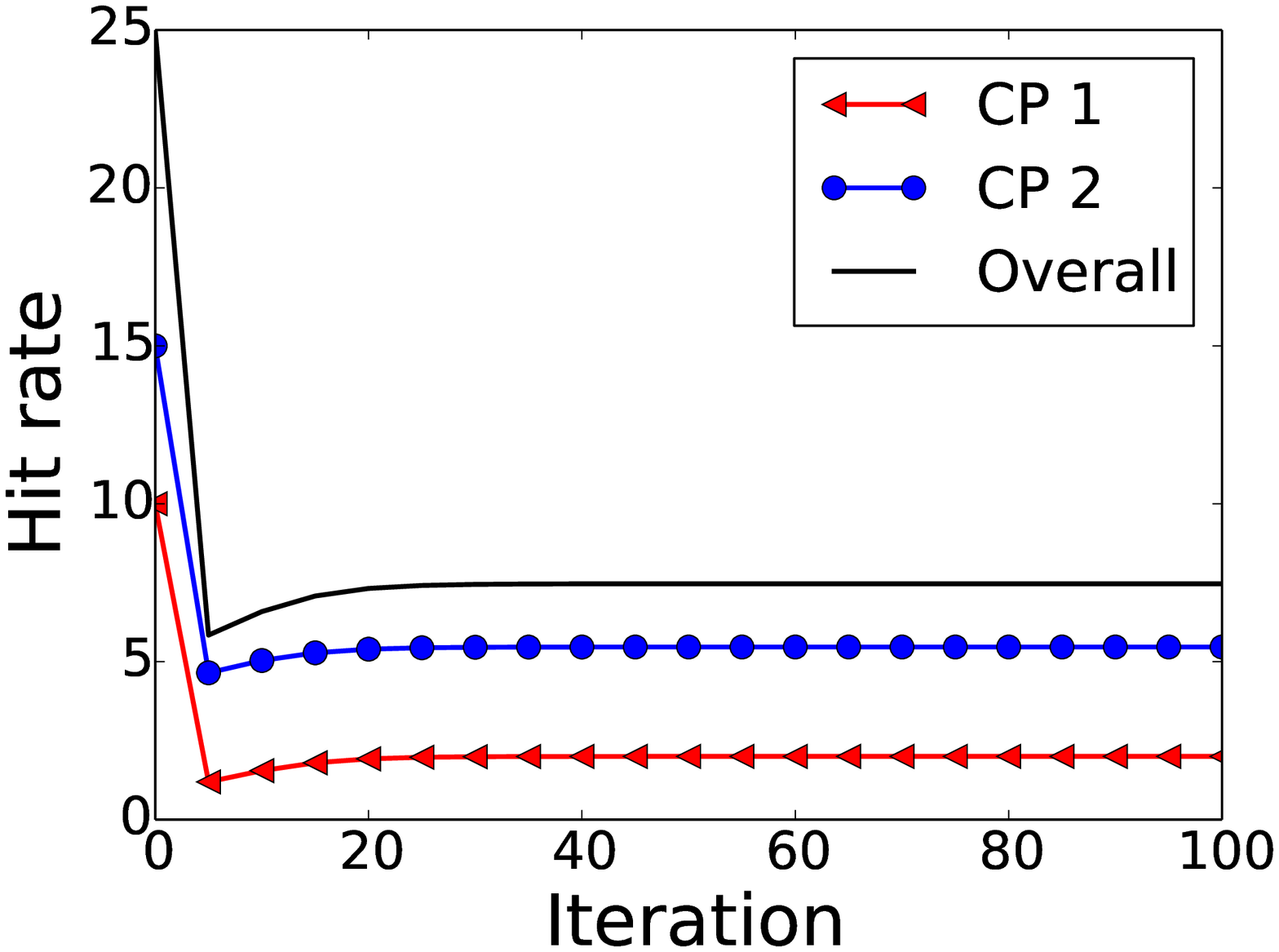}
    }
    \subfigure
    {
        \includegraphics[width=0.22\textwidth]{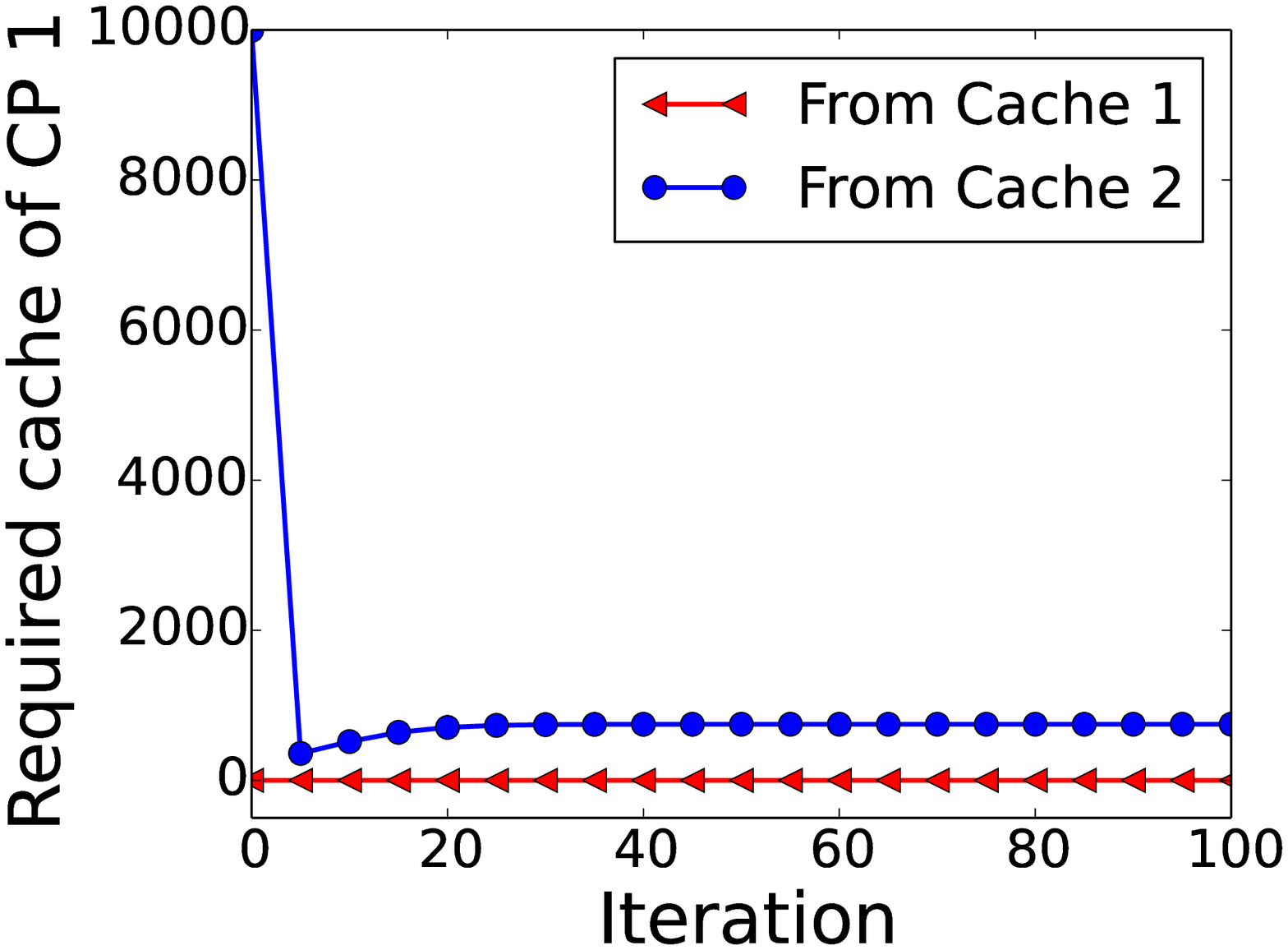}
    }
    \subfigure
    {
        \includegraphics[width=0.22\textwidth]{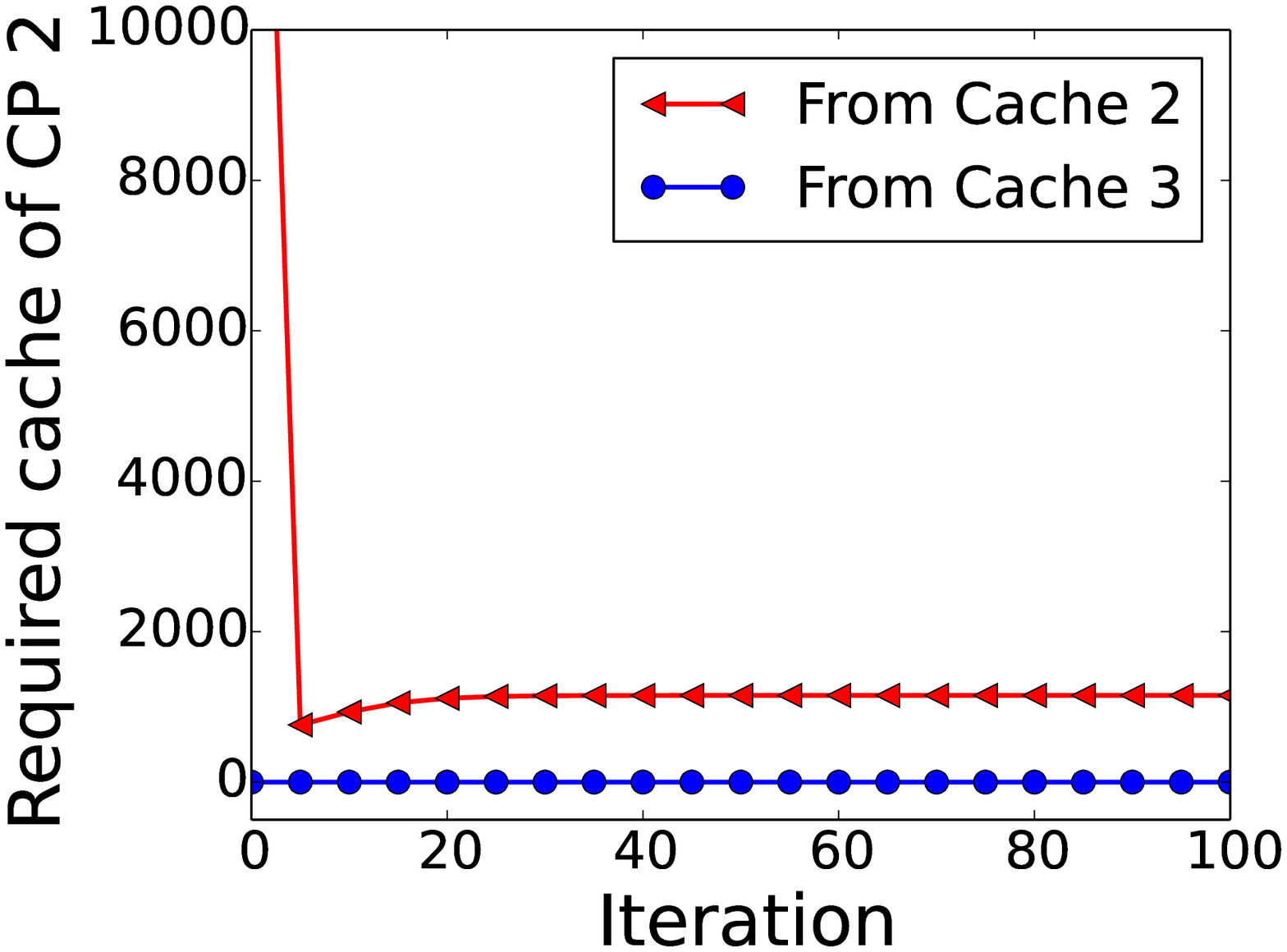}
    }
        \subfigure
    {
        \includegraphics[width=0.22\textwidth]{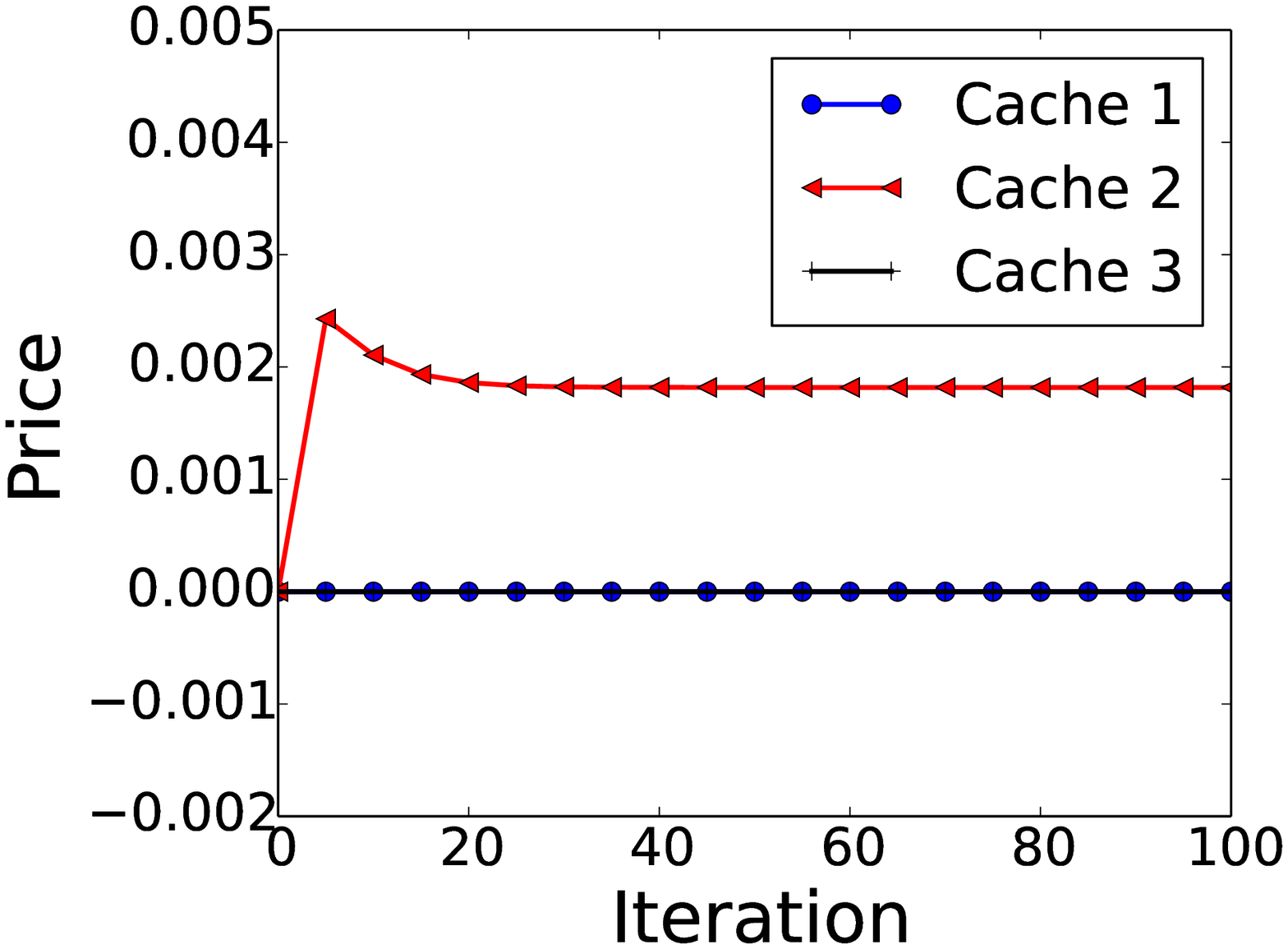}
    }\\
    \subfigure
    {
        \includegraphics[width=0.15\textwidth]{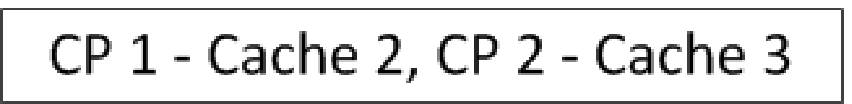}
    }\\
    \subfigure
    {
        \includegraphics[width=0.22\textwidth]{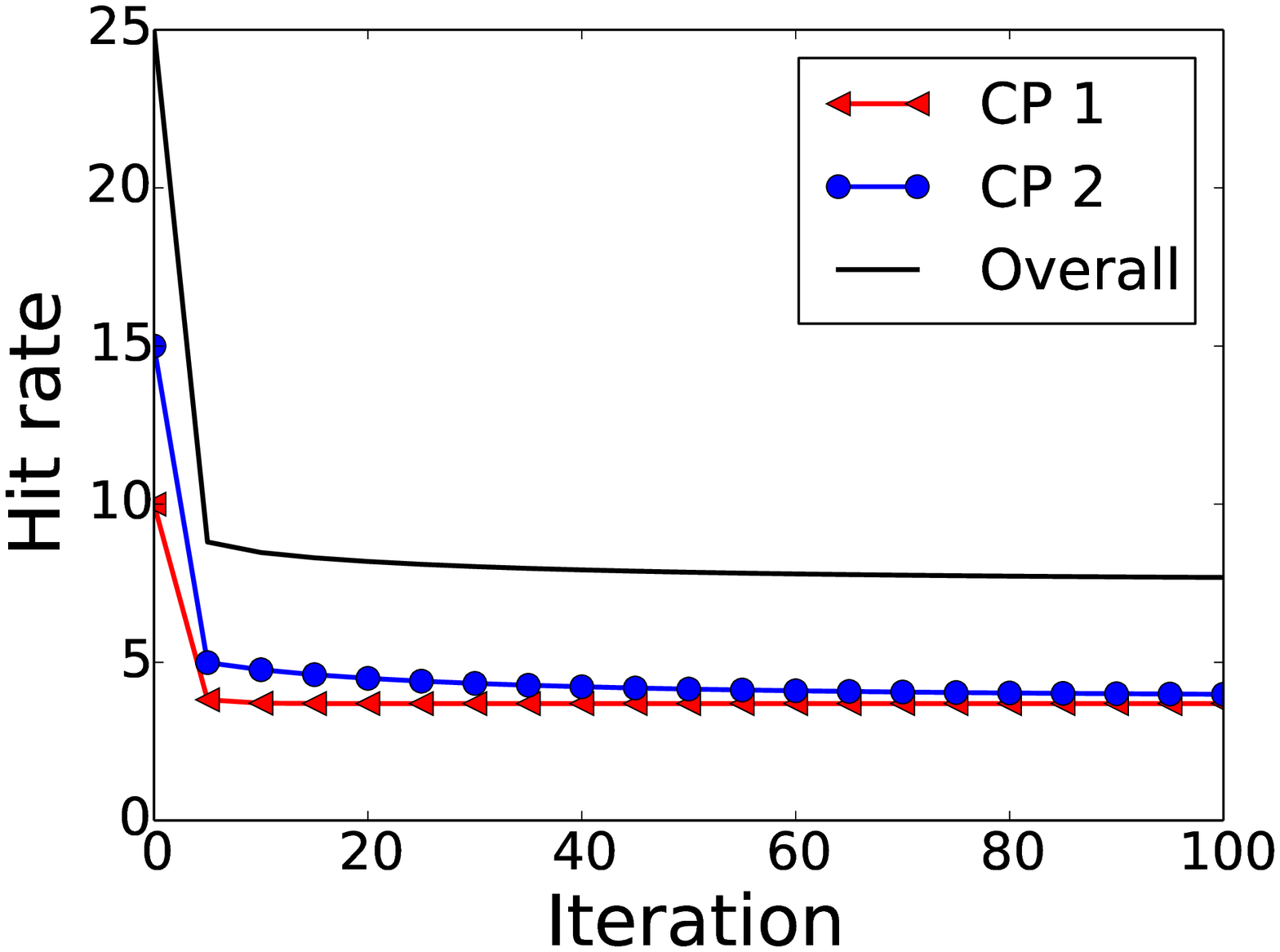}
    }
    \subfigure
    {
        \includegraphics[width=0.22\textwidth]{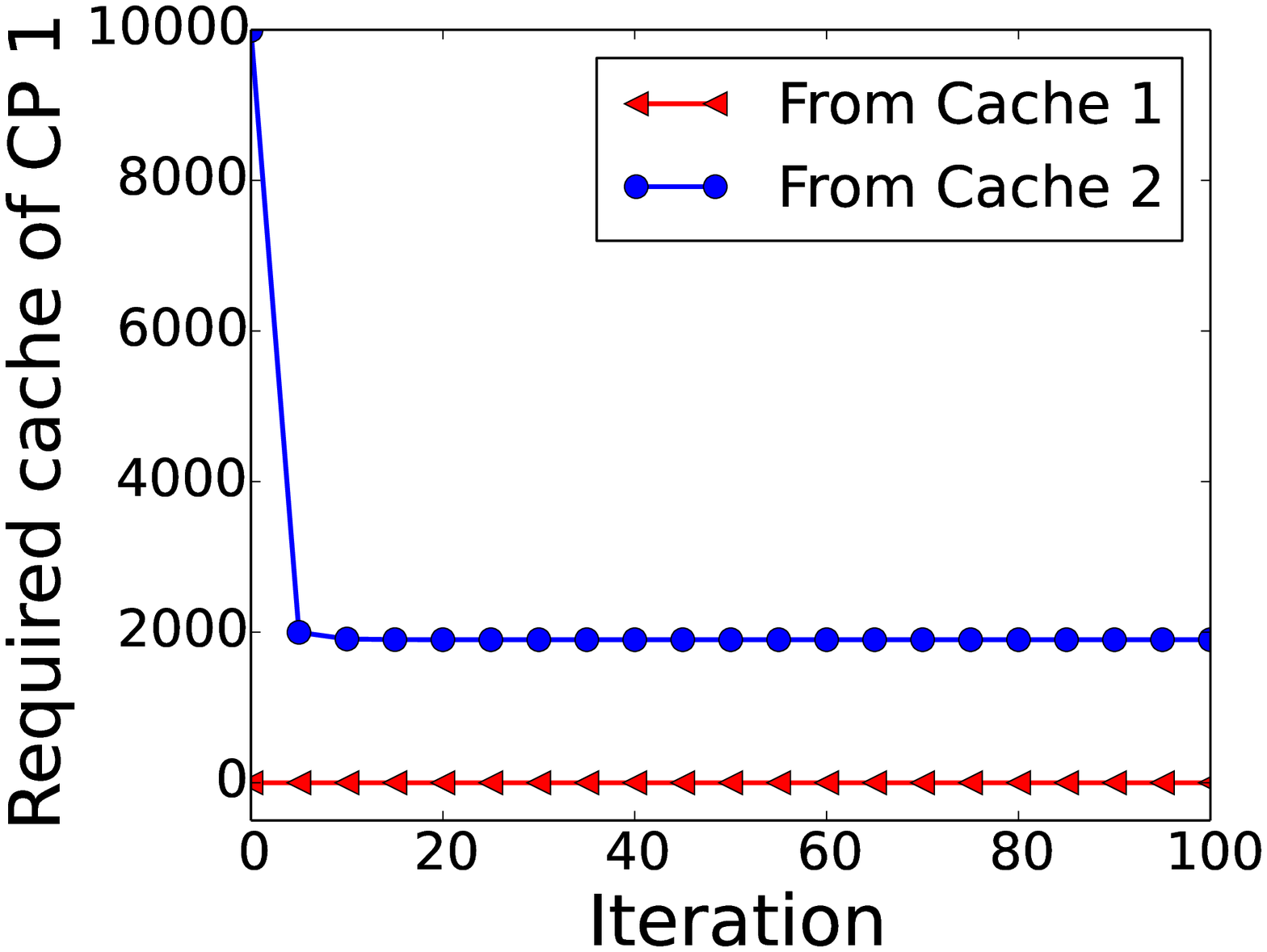}
    }
    \subfigure
    {
        \includegraphics[width=0.22\textwidth]{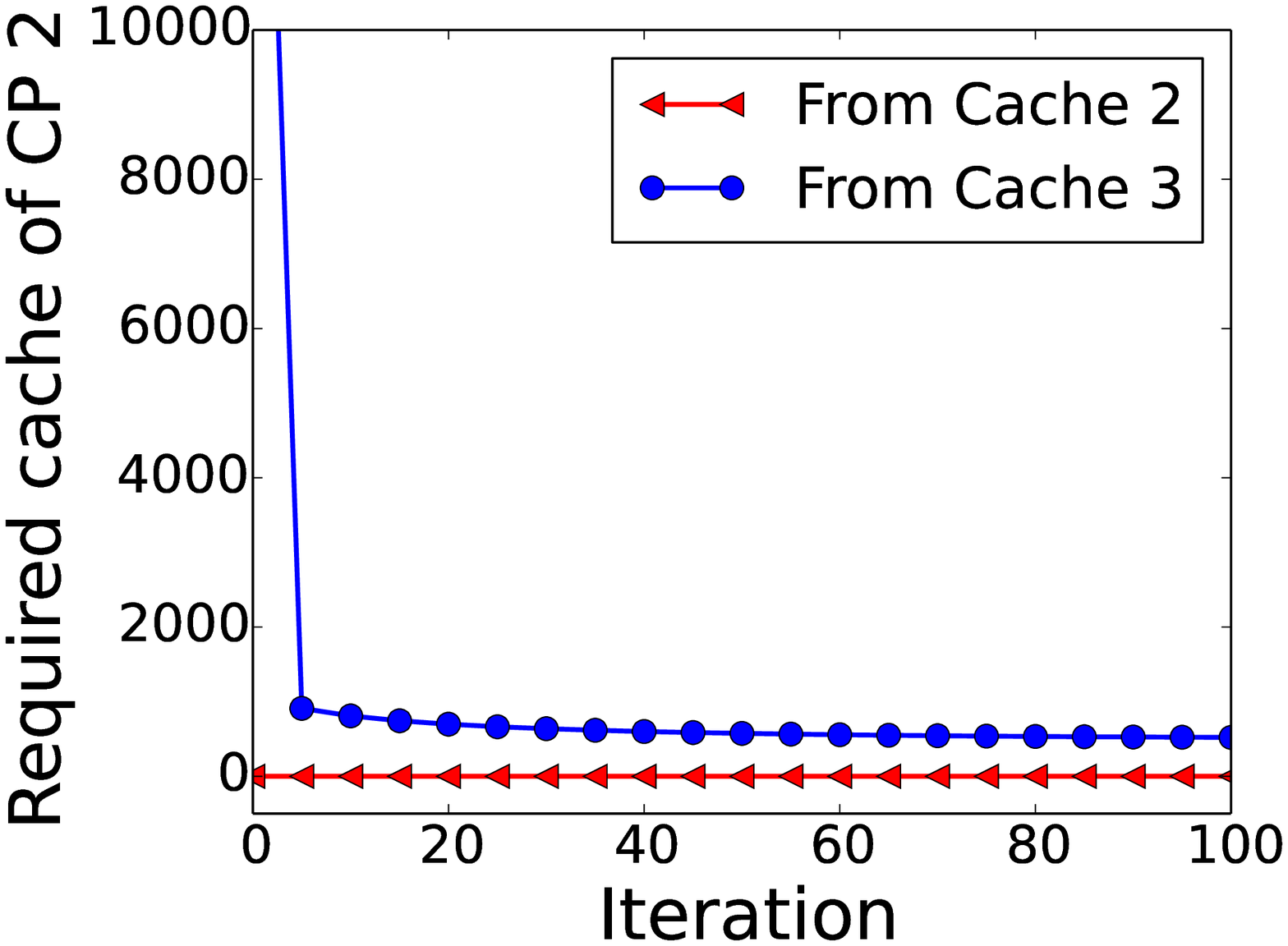}
    }
        \subfigure
    {
        \includegraphics[width=0.22\textwidth]{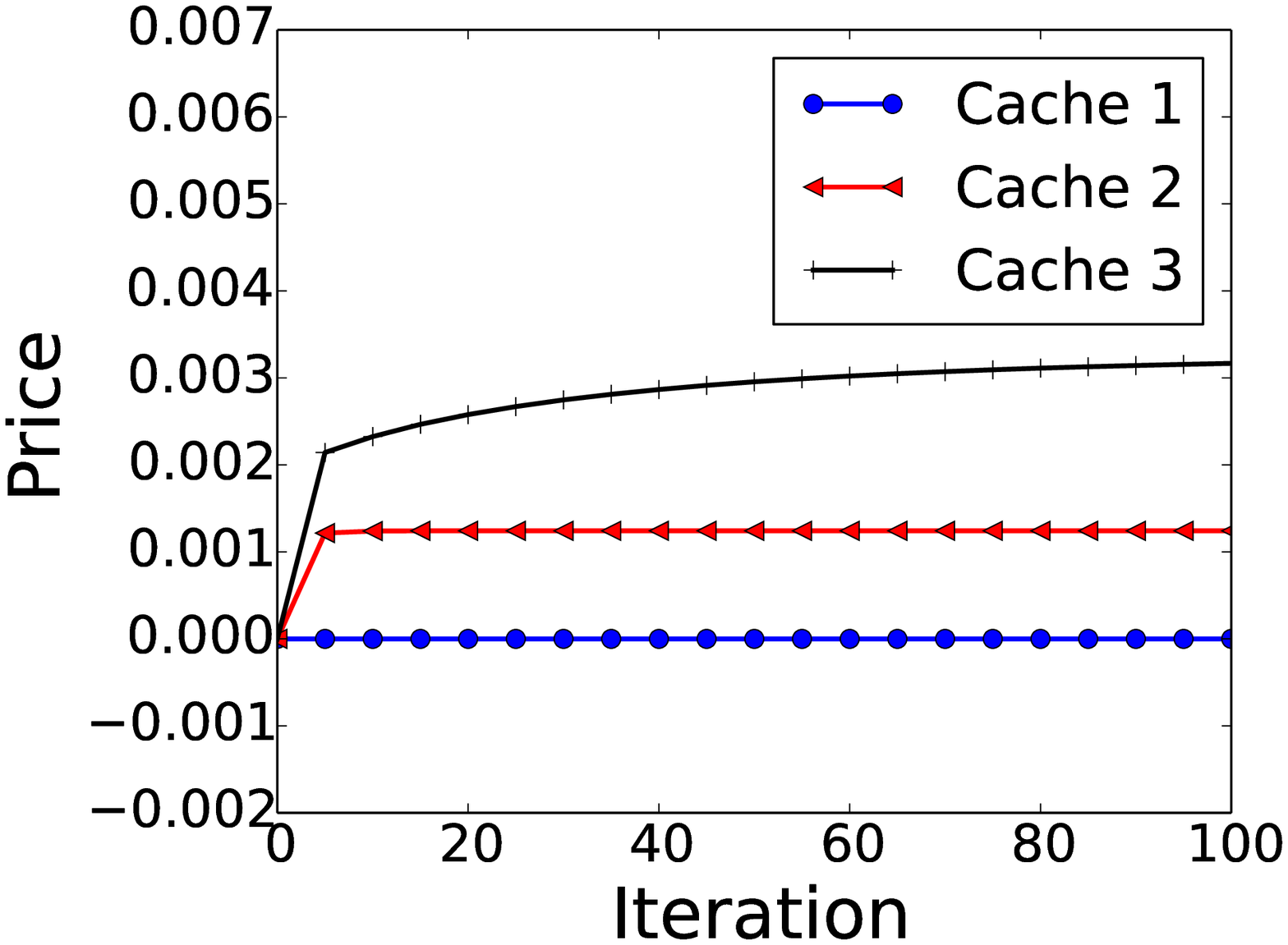}
    }
    \caption{Dynamics of decentralized algorithm.}
    \label{fig:dynamics}
\end{figure*}

\subsection{Bandwidth-constrained Scenario Evaluation}

For bandwidth-constrained scenario, we set the maximum volume between CP 1 and the two caches as $V_{11}=6$, $V_{12}=8$, and that of CP 2 as $V_{22}=10$ and $V_{23}=8$. Figure~\ref{fig:utility_comparison_bandwidthlimit} shows the observed utility as the capacity of Cache 2 varies. It can be seen that even with bandwidth limitation, our mechanism outperforms static routing in this case, i.e., when both content providers route their requests to the connected caches at equal probabilities and so the traffic between them does not exceed any volume limitation.

Figure~\ref{fig:traffic_distribution_bandwidthlimit} shows the optimal request routings for the two content providers, and it is interesting to observe that when Cache 2 is small ($C_2<500$), CP 1 routes its requests to Cache 1 at the maximum volume and CP 2 routes its requests to Cache 3 at the maximum volume. When Cache 2 becomes larger ($500<C_2<1500$), CP 2 begins to route its requests to Cache 2 at the maximum volume. When it continues to grow ($C_2>1500$), both providers route their requests to the shared cache at the maximum volume.

\begin{figure}
  \centering
  \includegraphics[width=0.28\textwidth]{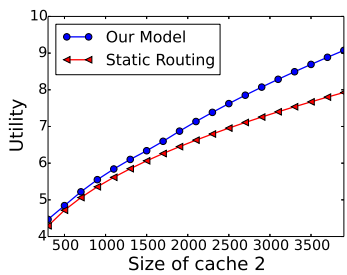}\\
  \caption{Efficacy of our model with bandwidth limitation as compared to static routing.}\label{fig:utility_comparison_bandwidthlimit}
\end{figure}

\begin{figure}
    \centering
    \subfigure[CP 1]
    {
        \includegraphics[width=0.225\textwidth]{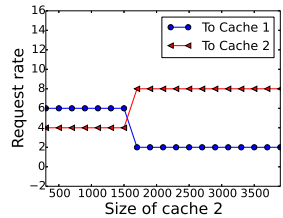}
        \label{fig:first_sub}
    }
    \subfigure[CP 2]
    {
        \includegraphics[width=0.225\textwidth]{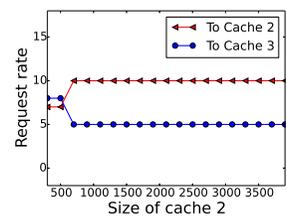}
        \label{fig:second_sub}
    }
    \caption{Traffic distribution of the two content providers with bandwidth limitation.}
    \label{fig:traffic_distribution_bandwidthlimit}
\end{figure}

%
%

Figure~\ref{fig:prices_bandwidth_limit} shows how the cache prices change under all four routings. Again we observe that our decentralized algorithm converges to the optimal solution.

\begin{figure*}[t]
    \centering
    \subfigure[CP 1-Cache 1, CP 2-Cache 2]
    {
        \includegraphics[width=0.22\textwidth]{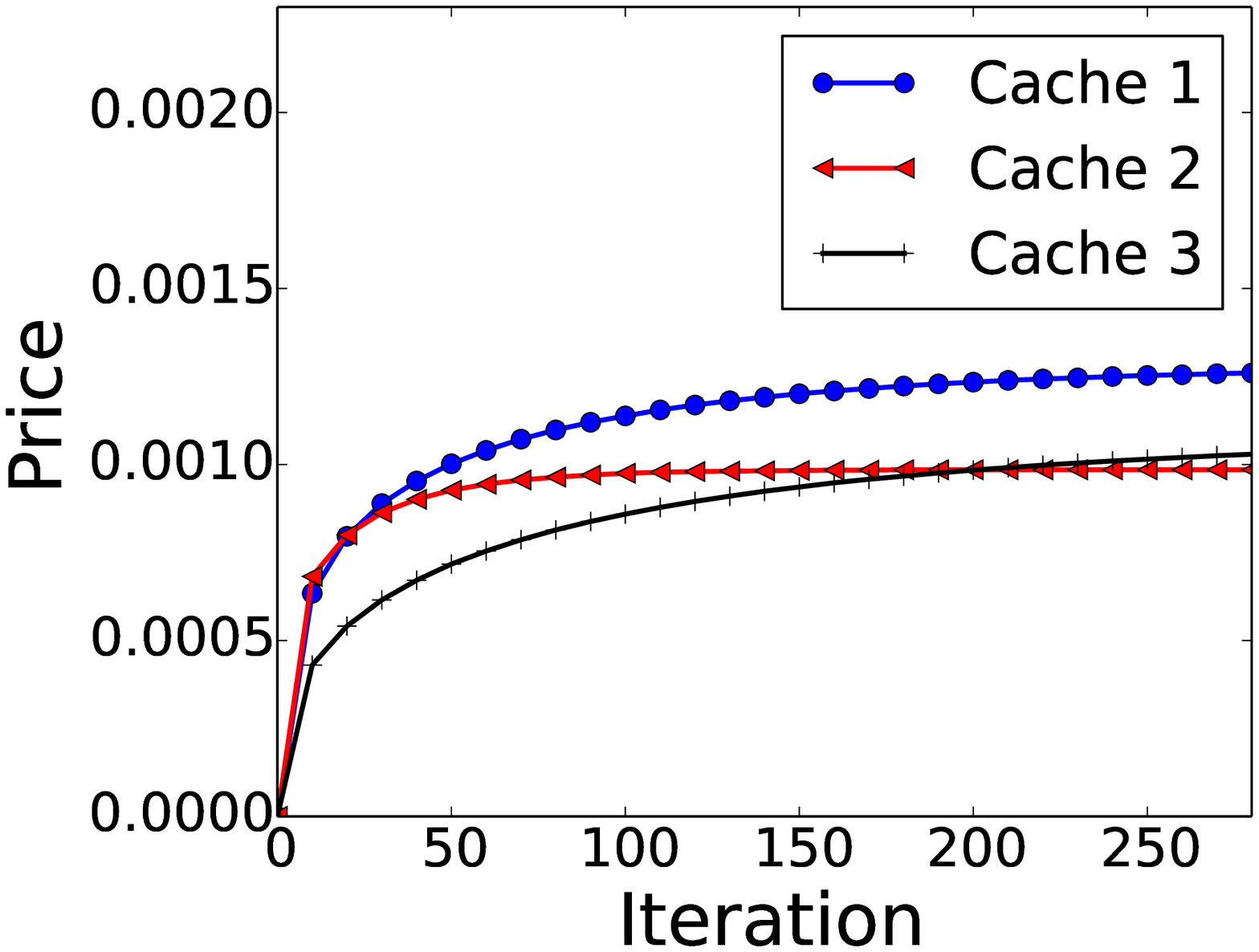}
    }
    \subfigure[CP 1-Cache 1, CP 2-Cache 3]
    {
        \includegraphics[width=0.22\textwidth]{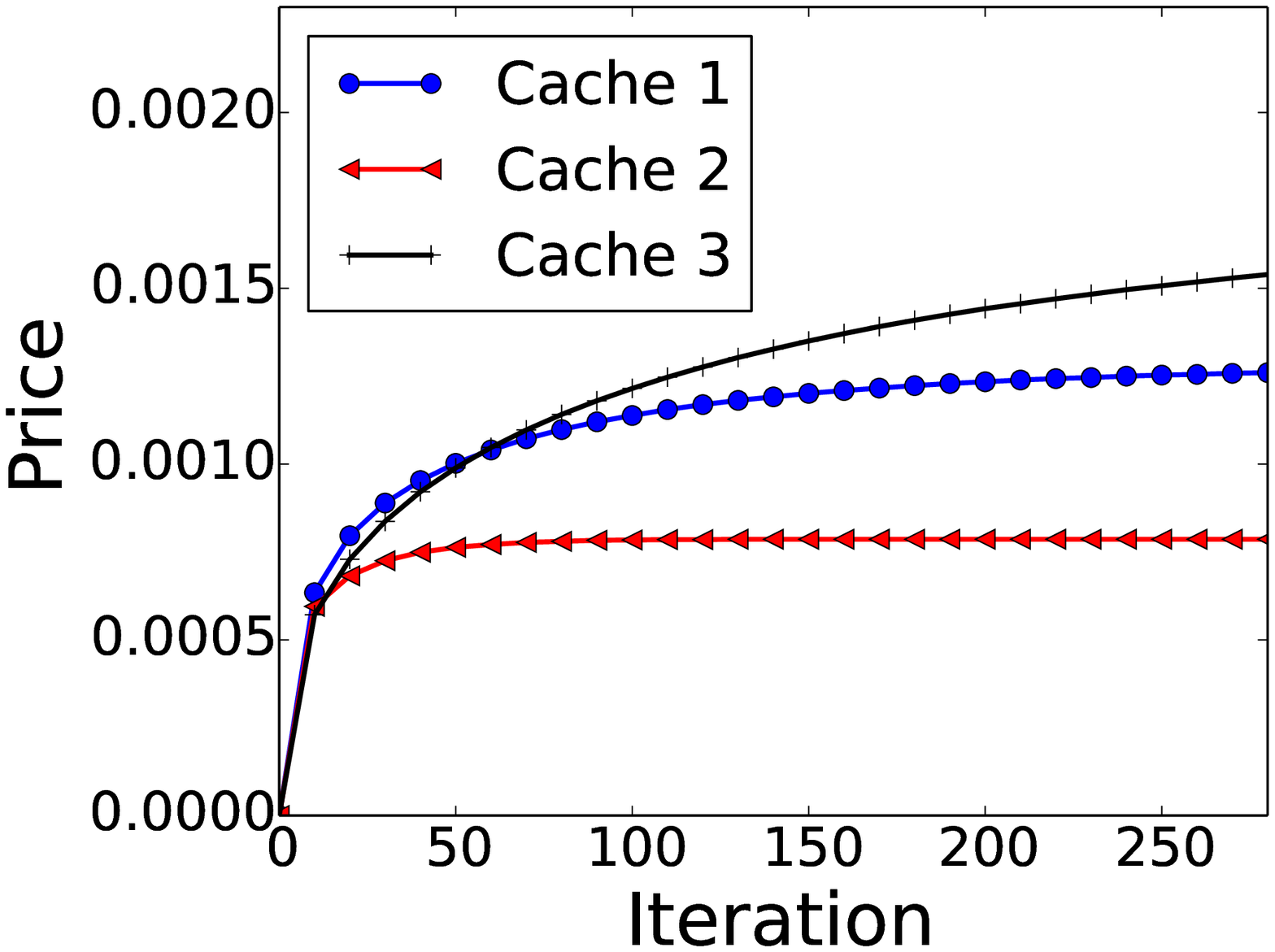}
    }
    \subfigure[CP 1-Cache 2, CP 2-Cache 2]
    {
        \includegraphics[width=0.22\textwidth]{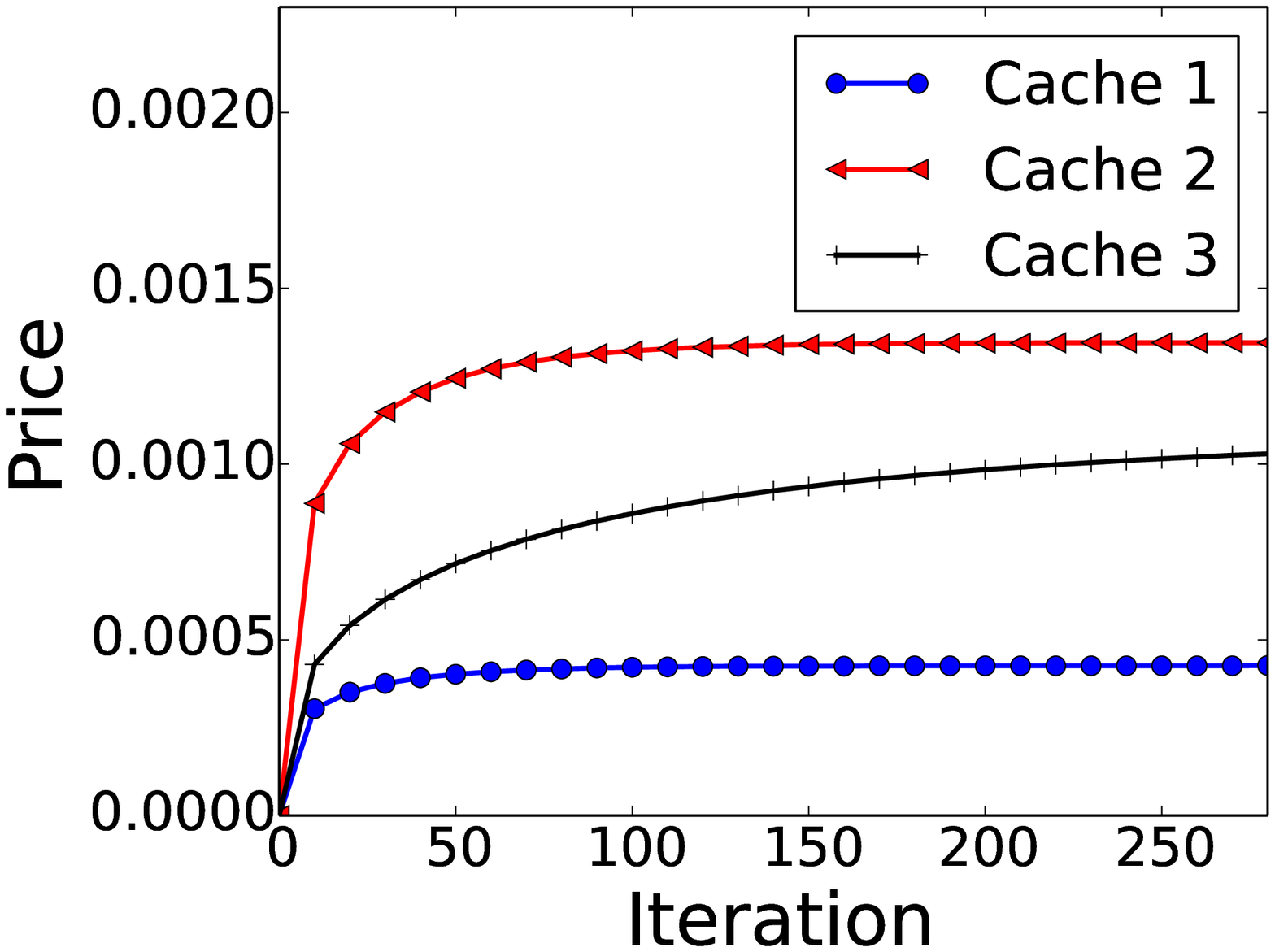}
    }
    \subfigure[CP 1-Cache 2, CP 2-Cache 3]
    {
        \includegraphics[width=0.22\textwidth]{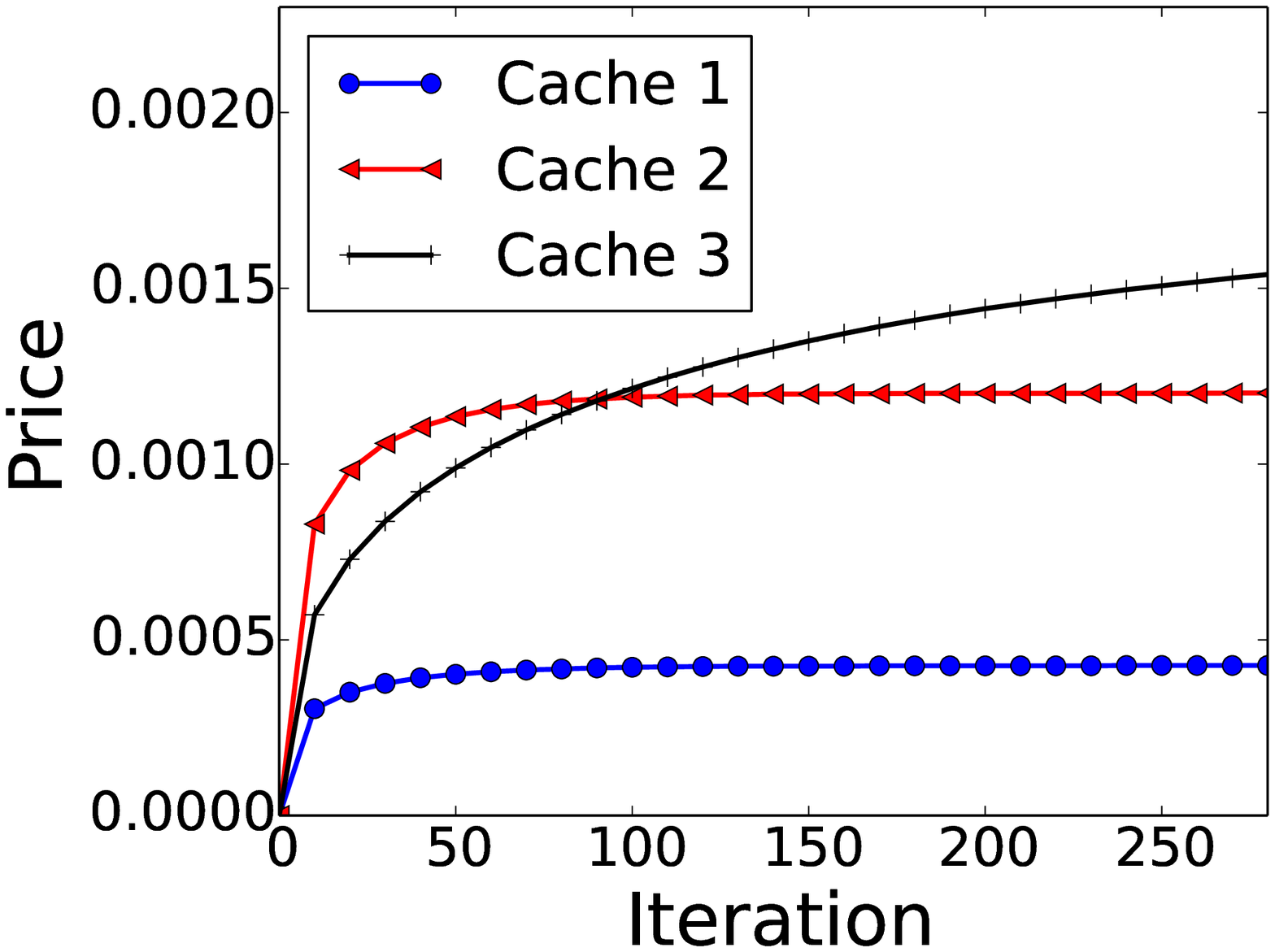}
    }
    \caption{Cache prices with bandwidth limitation under all four routings.}
    \label{fig:prices_bandwidth_limit}
\end{figure*}

We also perform numerical studies for the delay optimization scenario, and obtain similar results. To avoid redundancy and keep concise, we omit the corresponding figures. Nevertheless, we derive a consistent conclusion that our mechanism is efficient and the decentralized algorithms converge to the optimal solutions.

\subsection{Performance on More Complex Networks}
 
Here we report the performance of our mechanism on some more complex cache networks. In these networks, each CP connects to $2\scriptsize{\sim}5$ caches. Content population of CPs and their corresponding Zipf parameters are set as $1\times 10^3\scriptsize{\sim}5\times 10^3$ and $0.6\scriptsize{\sim} 0.8$, respectively. The request rate of CPs are chosen as $10\scriptsize{\sim}15$, and the size of caches as $200\scriptsize{\sim} 500$. All parameters follow a uniform distribution.

For each network setting with different number of CPs and caches, we conduct numerical experiment for 20 times, and give the average observed hit rate in Table 1. The results clearly illustrate that under all three different settings, our mechanism improves system performance by 30\% at least. Note that the joint optimization model is solved by SLSQP (sequential least squares programming) --- a nonlinear programming solver that returns local optimum for the general NLP problem. A higher performance gain is thus expected if global algorithm, e.g., {\it{GOP}}~\cite{floudas1990global}, is applied.  

\begin{table}[!hbp]
  \newcommand{\tabincell}[2]{\begin{tabular}{@{}#1@{}}#2\end{tabular}}
  \begin{tabular}{|c|c|c|c|c|}
    \hline
    \# Caches & \# CPs & \tabincell{c}{Aggre. Hitrate \\ by Stat. Routing} & \tabincell {c}{Aggre. Hitrate \\ by Joint Opt.} & \tabincell {c} {Perf. \\ Improvement}\\
    \hline
    3 & 5 & 24.8 & 32.7 & 31.8\% \\
    \hline
    5 & 10 & 48.6 & 63.4 & 30.4\% \\
    \hline
    8 & 20 & 86.8 & 114.5 & 31.5\% \\
    \hline
  \end{tabular}
  \caption{Performance on more complex networks}
\end{table}

Last, we emphasize that to evaluate our mechanism comprehensively it is better for us to  use real-life large-scale network topologies, such as the abilene network topology (9 routers, 26 links), the dtelekom network topology (68 routers, 546 links), the geant network topology (22 routers, 66 links), to name a few. The problem of using these real-life networks is that the number of connections between content providers and caches/routers are large and grows exponentially. In that case, it makes sense to restrict content provider---cache connections within a small subset, i.e. each content provider can only route requests to caches within the neighborhood of its end-users. In that case, it suffices to decompose the problem to a set of smaller sub-problems, each for a local network that is within the same order of the network that we considered in the numerical studies.


\section{Discussion}
\label{sec:discussion}

In this section, we explore some implications of our framework and present some future research directions.
  
\noindent {\bf{(1) Non equal-size disjoint content.}} Following the common practice, in this work we assume files are of equal sizes. However, files can be of variable sizes in real networks. In that case, we can divide each file into a number of small fixed-size chunks, and still treat these chunks as if they were independent disjoint files in large-scale caching systems, i.e., when the number of users accessing them are large enough. Nevertheless, this assumption needs to be carefully validated and its impact precisely measured. Another question is on the violation of the assumption of non-overlapping content, i.e., a video file can be served by multiple content providers. How to deal with these common files is another question that needs to be addressed before we apply our framework in such cases.
 
\noindent {\bf{(2) Adaptive online algorithms.}} Our framework provides a means to compute optimal cache partitioning and request routing. However, it is impractical to solve the optimization problems offline and then implement the optimal policy, since in real networks, system parameters can change over time, i.e., network traffic is generally unstable. As a result, adaptive algorithms are needed. Based on the fact that network topologies/connections are more stable than traffic, here we propose a two-layer adaptive algorithm. In the upper layer, the cache manager periodically measures the connectivity of CPs and caches, the corresponding latency of content delivery, and the statistics of requests from CPs. With these measurements, the cache manager decides to recalculate request routings if there are too much change in these statistics; in the lower-layer, the cache manager collects limited traffic information for each CP and tunes the cache partitioning for them under the given routing so as to adapt to traffic changes. Of course, it remains to develop/implement such dynamic online algorithms and evaluate them in real network environment, and we leave them for our future work. 

 \noindent {\bf{(3) Fairness.}} We can use different utility functions in the framework to implement different fairness among content providers. For example, the family of $\beta$-fair utility function expressed as $U_k(h_k)=\frac{h_{k}^{1-\beta}}{1-\beta}$ can be used to implement some interesting fairness.
When $\beta \to 1$, we have $U_k(h_k)=\log h_k$, and the goal is to implement proportional fairness;
when $\beta \to \infty$, it yields the objective $\text{max} \mathop \text{min}\limits_{k} h_k$, which corresponds to max-min fairness.
With such notions of fairness associated with utility functions, we actually provide a general and unified framework for implementing fair network resource allocation and request routing for different content providers.


\section{Conclusion}
\label{sec:conclusion}

In this paper, we study the problem of cache resource allocation in a Multi-CP Multi-Cache environment. We propose a joint cache partitioning and content-oblivious request routing scheme, and formulate an optimization problem in which the objective is to maximize network utilities through proper cache partitioning and request routing. We give the optimal request routing strategy for each content provider, establish the biconvexity property of the formulated problem, and further develop distributed (online) algorithms. We also extend our model to the bandwidth-constrained and delay optimization scenarios to show that it provides a general framework for cache resource management in a Multi-CP Multi-Cache network. From an economic perspective, we believe that our framework also helps to build a viable market model for in-network caching services.


\section*{Acknowledgements}
\label{sec:acknowledgements}

The work is supported by the National Natural Science Foundation of China (Grant No. 61502393), the Natural Science Basic Research Plan in Shaanxi Province of China (Grant No. 2017JM6066), the Fundamental Research Funds for the Central Universities (Grant No. 3102017zy031), and the National Science Foundation under grants CNS-1413998, CNS-1617437, CNS-1618339 and CNS-1617729. The work by John C.S. Lui is supported in part by the GRF 14200117.

\bibliographystyle{abbrv}
\bibliography{references}

\appendices

\section{Proof of Theorem~\ref{thrm:policy_concavity}}
\label{sec:proof_concavity}

We rely on the characteristic time approximation for LRU, FIFO and Random policies to prove that hit rate under these policies is a concave function of the cache size.

Based on equation~\eqref{eq:hitprob_lru} and~\eqref{eq:hitprob_fifo_random}, we can see that

\[\frac{\partial o(\lambda_i, T)}{\partial T} \ge 0 \quad \text{and} \quad \frac{\partial^2 o(\lambda_i, T)}{\partial T^2} \le 0\]

For a cache of size $C$, we have

\[C=\sum_i o(\lambda_i, T)\]

Taking derivatives with respect to $T$ from the two sides of the above equation, we get that

\[\frac{\partial C}{\partial T} \ge 0 \quad \text{and} \quad \frac{\partial^2 C}{\partial T^2} \le 0\]

The above inequalities imply that $C$ is a non-decreasing concave function of $T$ which then means $T$ is a non-decreasing convex function of $C$, i.e.,

\[\frac{\partial T}{\partial C} \ge 0 \quad \text{and} \quad \frac{\partial^2 T}{\partial C^2} \ge 0\]

Using the equation $C=\sum_i o(\lambda_i, T)$, and taking derivatives with respect to $C$ we obtain that

\[0 = \sum_{i}\frac{\partial^2 o(\lambda_i, T)}{\partial C^2}\]

We need to show

\[\sum_{i} \lambda_{i}\frac{\partial^2 o(\lambda_i, T)}{\partial C^2} \le 0\]

\begin{lemma}
  \label{lemma:concavity_proof_1}
Under LRU, FIFO and Random policies, for two files $f_i$ and $f_j$ such that $\lambda_i \ge \lambda_j$ ,$\frac{\partial^2 o(\lambda_j, T)}{\partial C^2}<0$ implies that $\frac{\partial^2 o(\lambda_i, T)}{\partial C^2}<0$.
\end{lemma}

\begin{proof}
  Starting with the equation~\eqref{eq:hitprob_lru} and~\eqref{eq:hitprob_fifo_random} and taking derivatives with respect to $C$ we obtain
  \begin{equation}
    \label{eq:concavity_proof_eq1}
\frac{\partial^2 o(\lambda_i, T)}{\partial C^2}=-\lambda_i^2 e^{-\lambda_iT}(\frac{\partial T}{\partial C})^2+\lambda_i e^{-\lambda_iT}\frac{\partial ^2 T}{\partial C^2},  
  \end{equation}
  
\noindent for LRU policy, and 
\begin{equation}
  \label{eq:concavity_proof_eq2}
\frac{\partial^2 o(\lambda_i, T)}{\partial C^2}=-\frac{2\lambda_i^2}{(1+\lambda_i T)^3}(\frac{\partial T}{\partial C})^2+ \frac{\lambda_i}{(1+\lambda_i T)^2}\frac{\partial ^2 T}{\partial C^2},  
\end{equation}

\noindent for FIFO and Random policies.

For LRU policy~\eqref{eq:concavity_proof_eq1} implies that

\[\frac{\partial^2 o(\lambda_i, T)}{\partial C^2}<0 \quad \text{if} \quad \frac{\partial^2 T}{\partial C^2}<\lambda_i(\frac{\partial T}{\partial C})^2\]

\noindent On the other hand, 

 \[\frac{\partial^2 T}{\partial C^2}<\lambda_j(\frac{\partial T}{\partial C})^2 \quad \text{if} \quad \frac{\partial^2 T}{\partial C^2}<\lambda_i(\frac{\partial T}{\partial C})^2,\]

\noindent since $\lambda_j \le \lambda_i$ and hence

\[\frac{\partial^2 o(\lambda_i, T)}{\partial C^2}<0 \quad \text{if} \quad \frac{\partial^2 o(\lambda_j, T)}{\partial C^2}<0.\]

Similarly, for FIFO and Random policies,~\eqref{eq:concavity_proof_eq2} implies that

\[\frac{\partial^2 o(\lambda_j, T)}{\partial C^2}<0 \quad \text{if} \quad \frac{\partial^2 T}{\partial C^2}<\frac{2\lambda_j}{1+\lambda_j T}(\frac{\partial T}{\partial C})^2,\]

\noindent since $\lambda_j \leq \lambda_i$ and $\frac{2\lambda}{1+\lambda T}$ is an increasing function of $\lambda$. Therefore

\[\frac{\partial^2 o(\lambda_i, T)}{\partial C^2}<0 \quad \text{if} \quad \frac{\partial^2 o(\lambda_j, T)}{\partial C^2}<0.\]

\end{proof}

\begin{lemma}
  \label{lemma:concavity_proof_2}
  Under LRU, FIFO and Random policies, for two files $f_i$ and $f_j$ such that $\lambda_i \geq \lambda_j$ ,$ \frac{\partial^2 o(\lambda_i, T)}{\partial C^2}>0$ implies that $\frac{\partial^2 o(\lambda_j, T)}{\partial C^2}>0$.
\end{lemma}

Lemma~\ref{lemma:concavity_proof_2} can be proved with similar steps as in proof of Lemma~\ref{lemma:concavity_proof_1}.

Now, starting with the equation

\[C=\sum_{i}o(\lambda_i, T),\]

\noindent and taking derivative with respect to $C$ from both sides of the equation, we obtain

\[0=\frac{\partial^2 o(\lambda_i, T) }{\partial C^2}\]

Cache hit rate is expressed as

\[h=\sum_i \lambda_i o(\lambda_i, T)\]

\noindent and taking the derivative with respect to cache size from both sides of the equation we obtain

\begin{equation}
  \label{eq:hitrate_first_derivative}
  \frac{\partial h}{\partial C}=\sum_i \lambda_i\frac{\partial o(\lambda_i, T)}{\partial C}\ge 0
\end{equation}

\noindent Taking a second derivative yields

\[\frac{\partial ^2 h}{\partial C^2}=\sum_{i}\lambda_i \frac{\partial ^2 o(\lambda_i, T)}{\partial C^2}.\]

\noindent Since $\sum_{i} \frac{\partial ^2 o(\lambda_i, T)}{\partial C^2}=0$, Lemma~\ref{lemma:concavity_proof_1} and~\ref{lemma:concavity_proof_2} imply that

\begin{equation}
  \label{eq:hitrate_second_derivative}
\frac{\partial ^2 h}{\partial C^2}=\sum_{i}\lambda_i \frac{\partial ^2 o(\lambda_i, T)}{\partial C^2} \leq 0  
\end{equation}

From~\eqref{eq:hitrate_first_derivative} and~\eqref{eq:hitrate_second_derivative}, we conclude that hit rate is an increasing concave function of the cache size.

\section{Proof of Theorem~\ref{thrm:no_splitting}}
\label{sec:proof_thm1}
We have $h_k=\sum_{m=1}^{M}\sum_{i=1}^{N_k} {\lambda_{ik}p_{km}}o_{ik}(\lambda_{ik},C_{km})$. Since $o_{ik}(\lambda_{ik},C_{km})$ is increasing in $C_{km}$ and $\lambda_{ik}$ is constant, it can be seen that $h_k$ is increased by moving requests from one cache with a small $o_{ik}$ to the cache with a larger one, i.e., by moving requests from cache $m$ to cache $s$ if $o_{ik}(\lambda_{ik},C_{km}) < o_{ik}(\lambda_{ik},C_{ks})$. The largest value of $h_k$ can be obtained by routing all requests of CP $k$ to the cache with the largest $o_{ik}$, or equivalently, to the cache from which CP $k$ is allocated the largest slice.

\section{Proof of Theorem~\ref{thrm:concavity}}
\label{sec:proof_thm2}
Given a routing configuration, let $CP(m)$ be the set of CPs that route requests to cache $m$. When all requests of each CP are directed to one cache it connects, we can reformulate problem~\eqref{eq:multi_cache_model_reformulated} as:

\[ \text{maximize } \sum_{m=1}^M \sum_{k \in CP(m)}w_kU_{k}(h_{km}(C_{km},1)) \]

\begin{equation}\label{eq:reformulate_problem}
\begin{aligned}
  \sum_{k \in CP(m)} C_{km} \leq C_m,  \quad  \forall m \\
\end{aligned}
\end{equation}

Since the objective function of the above problem is separable, and the variables are not coupled, problem~\eqref{eq:reformulate_problem} can be further decomposed into $M$ subproblems, one for each cache. For cache $m$, the subproblem becomes:

\[ \text{maximize } \sum_{k \in CP(m)}w_kU_{k}(h_{km}(C_{km},1)) \]

\begin{equation}\label{eq:reformulate_subproblem}
\begin{aligned}
  \sum_{k \in CP(m)} C_{km} \leq C_m \\
\end{aligned}
\end{equation}

\begin{lemma}
\label{lemma:concavity}
The single-cache resource allocation problem~\eqref{eq:reformulate_subproblem} has a unique optimal solution.
\end{lemma}

\begin{proof}
Since $h_{km}$ is concave in $C_{km}$ and $U_{k}$ is also concave, the objective function
in \eqref{eq:reformulate_subproblem} is concave. Furthermore, as the feasible solution set is convex, a unique maximizer, called the optimal solution, exists.

\end{proof}

\section{Proof of Theorem~\ref{thrm:np_hardness}}
\label{sec:proof_NP_hardness}
We have proved that the optimal solution is such that each CP directs all requests to one cache. Now if we regard bins as caches, items as CPs, the budget of bins as the capacity of caches, the weight of items as the allocated cache amount to CPs, and the profit of each item as the resulting utility of each CP, then we can see that the Separable Assignment Problem (SAP~\cite{fleischer2006tight}) is a special case of Problem~\eqref{eq:multi_cache_model_reformulated} (note that the weights and profits in the considered problem are not fixed, while in SAP they are fixed). SAP is a general class of maximum assignment problem with packing constraints, and it has been proved to be NP-complete since the knapsack problem is its special case~\cite{fleischer2006tight}.


\end{document}